\documentclass[10pt]{article}

\usepackage[utf8]{inputenc}
\usepackage[T1]{fontenc}

\usepackage{epsf}
\usepackage{amsmath}

\allowdisplaybreaks

\usepackage[showframe=false]{geometry}
\usepackage{changepage}

\usepackage{epsfig}
\usepackage{amssymb}

\usepackage{amsthm}
\usepackage{setspace}
\usepackage{cite}
\usepackage{mcite}

\usepackage{algorithmic}  
\usepackage{algorithm}

\usepackage{shadow}
\usepackage{fancybox}
\usepackage{fancyhdr}

\usepackage{color}
\usepackage[usenames,dvipsnames,svgnames,table]{xcolor}
\newcommand{\bl}[1]{\textcolor{blue}{#1}}

\definecolor{mypurple}{rgb}{.4,.0,.5}

\usepackage[hyphens]{url}

\usepackage[colorlinks=true,
            linkcolor=black,
            urlcolor=blue,
            citecolor=purple]{hyperref}

\usepackage{breakurl}

\def\y{{\bf y}}

\def\x{{\bf x}}

\def\x{{\mathbf x}}

\def\u{{\bf u}}

\def\x{{\bf x}}
\def\y{{\bf y}}

\def\q{{\bf q}}
\def\m{{\bf m}}

\def\c{{\bf c}}

\def\h{{\bf h}}

\def\cH{{\mathcal H}}

\def\be{\begin{equation}}
\def\ee{\end{equation}}
\def\ba{\left[\begin{array}}
\def\ea{\end{array}\right]}

\def\u{{\bf u}}

\def\x{{\bf x}}
\def\y{{\bf y}}

\def\q{{\bf q}}

\def\c{{\bf c}}

\def\1{{\bf 1}}

\def\0{{\bf 0}}

\def\erf{\mbox{erf}}
\def\erfc{\mbox{erfc}}

\def\calX{{\cal X}}

\def\mR{{\mathbb R}}
\def\mN{{\mathbb N}}
\def\mE{{\mathbb E}}

\def\mB{{\mathbb B}}

\def\lp{\left (}
\def\rp{\right )}

\def\y{{\bf y}}

\def\x{{\bf x}}

\def\x{{\mathbf x}}

\def\u{{\bf u}}

\def\x{{\bf x}}
\def\y{{\bf y}}

\def\q{{\bf q}}

\def\c{{\bf c}}

\def\h{{\bf h}}

\def\cH{{\cal H}}

\def\be{\begin{equation}}
\def\ee{\end{equation}}
\def\ba{\left[\begin{array}}
\def\ea{\end{array}\right]}

\def\u{{\bf u}}

\def\x{{\bf x}}
\def\y{{\bf y}}

\def\q{{\bf q}}

\def\c{{\bf c}}

\def\({\left (}
\def\){\right )}

\def\1{{\bf 1}}
\def\m{{\bf m}}
\def\q{{\bf q}}

\def\0{{\bf 0}}

\def\cX{{\mathcal X}}

\usepackage{xcolor}
\usepackage{color}

\definecolor{darkgreen}{rgb}{0, 0.4,0}

\definecolor{purplebrown}{rgb}{0.5,0.1,0.6}

\definecolor{ultclupcol}{rgb}{0.1,0.5,0.5}

\definecolor{mytrycolor}{rgb}{0.5,0.7,0.2}

\definecolor{ultclupcola}{rgb}{.5,0,.5}

\definecolor{shadebrown}{rgb}{0.1,0.1,0.9}
\definecolor{lightblue}{rgb}{0.2,0,1}

\usepackage{fancybox}
\usepackage{graphicx}
\usepackage{epstopdf}
\usepackage{epsfig}
\usepackage{wrapfig}
\usepackage{subfigure}

\usepackage{xcolor}
\usepackage{tcolorbox}

\newtcbox{\xmybox}{on line,
arc=7pt,
before upper={\rule[-3pt]{0pt}{10pt}},boxrule=0pt,
boxsep=0pt,left=6pt,right=6pt,top=0pt,bottom=0pt,enhanced, coltext=blue, colback=white!10!yellow}

\newtcbox{\xmyboxa}{on line,
arc=7pt,
before upper={\rule[-3pt]{0pt}{10pt}},boxrule=0pt,
boxsep=0pt,left=6pt,right=6pt,top=0pt,bottom=0pt,enhanced, colback=white!10!yellow}

\newtcbox{\xmyboxb}{on line,
arc=7pt,
before upper={\rule[-3pt]{0pt}{10pt}},boxrule=1pt,colframe=darkgreen!100!blue,
boxsep=0pt,left=6pt,right=6pt,top=0pt,bottom=0pt,enhanced, colback=white!10!yellow}

\newtcbox{\xmyboxc}{on line,
arc=7pt,
before upper={\rule[-3pt]{0pt}{10pt}},boxrule=.7pt,colframe=blue!100!blue,
boxsep=0pt,left=6pt,right=6pt,top=0pt,bottom=0pt,enhanced, coltext=blue, colback=white!10!yellow}

\newtcbox{\xmytboxa}{on line,
arc=7pt,
before upper={\rule[-3pt]{0pt}{10pt}},boxrule=.0pt,colframe=pink!50!yellow,
boxsep=0pt,left=6pt,right=6pt,top=0pt,bottom=0pt,enhanced, coltext=white, colback=blue!40!red}

\newtcbox{\xmytboxb}{on line,
arc=7pt,
before upper={\rule[-3pt]{0pt}{10pt}},boxrule=.0pt,colframe=pink!50!yellow,
boxsep=0pt,left=6pt,right=6pt,top=0pt,bottom=0pt,enhanced, coltext=white, colback=white!40!green}

\makeatletter
\newcommand\subsubsubsection{\@startsection{paragraph}{4}{\z@}{-2.5ex\@plus -1ex \@minus -.25ex}{1.25ex \@plus .25ex}{\normalfont\normalsize\bfseries}}
\newcommand\subsubsubsubsection{\@startsection{subparagraph}{5}{\z@}{-2.5ex\@plus -1ex \@minus -.25ex}{1.25ex \@plus .25ex}{\normalfont\normalsize\bfseries}}
\makeatother

\newtheorem{theorem}{Theorem}

\begin{document}

\begin{singlespace}

\title {A CLuP algorithm to practically achieve $\sim 0.76$ SK--model ground state free energy 
}
\author{
\textsc{Mihailo Stojnic
\footnote{e-mail: {\tt flatoyer@gmail.com}} }}
\date{}
\maketitle

\centerline{{\bf Abstract}} \vspace*{0.1in}

We consider algorithmic determination of the $n$-dimensional  Sherrington-Kirkpatrick (SK) spin glass model ground state free energy. It corresponds to  a binary maximization of an indefinite quadratic form and under the \emph{worst case} principles of the classical NP complexity theory it is hard to approximate within a $\log(n)^{const.}$ factor. On the other hand, the SK's random nature allows (polynomial) spectral methods to \emph{typically} approach the optimum within a constant factor. Naturally one is left with the fundamental question: can the residual (constant) \emph{computational gap} be erased?

Following the success of \emph{Controlled Loosening-up} (CLuP) algorithms in planted models, we here devise a simple practical CLuP-SK algorithmic procedure for (non-planted) SK models. To analyze the \emph{typical} success of the algorithm we associate to it (random) CLuP-SK  models. Further connecting to recent  random processes  studies \cite{Stojnicnflgscompyx23,Stojnicsflgscompyx23}, we characterize the models and CLuP-SK algorithm
via  fully lifted random duality theory (fl RDT) \cite{Stojnicflrdt23}. Moreover, running the algorithm we demonstrate  that its performance is in an excellent agrement with theoretical predictions. In particular, already for $n$  on the order of a few thousands CLuP-SK achieves $\sim 0.76$ ground state free energy and remarkably closely approaches theoretical $n\rightarrow\infty$ limit $\approx 0.763$. For all practical purposes,  this renders computing SK model's near ground state free energy as a \emph{typically} easy problem.

\vspace*{0.25in} \noindent {\bf Index Terms: SK model; CLuP algorithm; Fully lifted random duality theory}.

\end{singlespace}

\section{Introduction}
\label{sec:back}

For a given integer $n\in\mN$, consider the following set of vertices of an $n$-dimensional ``binary'' cube
\begin{eqnarray}\label{eq:inteq2}
\mB^n \triangleq \left \{ \x| \x\in\mR^n, \x_i^2=\frac{1}{n}\right \}.
\end{eqnarray}
Taking $G\in\mR^{n\times n}$ we are interested in the optimization of the associated (indefinite) quadratic form
\begin{eqnarray}\label{eq:inteq1}
 \max_{\x\in\mB^n} \x^TG\x.
\end{eqnarray}
Within the classical NP complexity theory, this is one of the foundational problems. In such a context, it is even hard to approximate it within a $\log(n)^{const.}$ factor \cite{AroraBKSH05} (see, also \cite{CharikarW04,Meg01} for further algorithmic considerations including integer constraints relaxations). While NP concepts are widely believed to be true, they often remain powerless when it comes to proper assessment of \emph{typical} solvability. The reason is their \emph{worst case} (usually deterministic) nature which predominantly relies on instances that may be far away from \emph{typical} problem structures.

Throughout the paper, we focus on a \emph{random} variant of (\ref{eq:inteq1}), where $G$ is comprised of independent standard normals. This variant corresponds to the ground state energy of the celebrated Sherrington-Kirkpatrick (SK) model. In the mid seventies of the last century, the SK model was proposed in  \cite{SheKir72}  as  a long range spin interactions antipode to the nearest neighbor (Edwards-Anderson) spin-glass model  \cite{EdwAnd75}. Intensive studying over several ensuing decades uncovered many of its remarkable properties \cite{Par79,Par80,Par83,Parisi80,SheKir72,Talbook11a,Talbook11b,Pan10,Pan10a,Pan13,Pan13a,Pan14,Guerra03,Tal06}.  Beyond fascinating, exactly half a century later, it still remains among the handful most widely studied statistical mechanics concepts  with applications in a plethora od scientific fields ranging from physics, mathematics, and probability \cite{Talbook11a,Talbook11b,Pan10,Pan10a,Pan13,Pan13a,Pan14,Guerra03,Tal06,StojnicRegRndDlt10,Stojnicgscompyx16,Stojnicgscomp16,SchTir03},  to information theory and signal processing \cite{BayMon10,BayMon10lasso,StojnicCSetam09,StojnicGenLasso10}, and reaching to computer science, machine leaning, and neural networks \cite{Moll12,CojOgh14,MezParZec02,MerMezZec06,FriWor05,AchPer04,AlmSor02,Wast09,LinWa04,Talbook11a,Talbook11b,NaiPraSha05,CopSor99,Ald01}.

Many of the applications are seemingly even unrelated to SK models as the underlying problems do not have the form given in (\ref{eq:inteq1}). It is however, the insights gained through studies and applications of SK models that inspired  such applications. One typically utilizes  statistical/experimental physics postulates to uncover and explain natural phenomena and by doing so provides valuable intuitions for mathematically rigorous studies and concrete practical implementations. This fruitful synergistic approach repeats itself as a common denominator in the above mentioned areas allowing even further branching out and studying many different topics within each of them. In a pool of diverse examples we mention  spherical \cite{GarDer88,Gar88,SchTir03,StojnicGardGen13,Tal05,Talbook11a,Talbook11b,Wendel,Winder,Winder61,Cover65,DonTan09Univ,StojnicGardSphNeg13,FraPar16,FraSclUrb19,AlaSel20}
 and binary \cite{GarDer88,Gar88,BaldMPZ23,StojnicGardGen13,BMPZ23,AbbLiSly21b,PerkXu21,AbbLiSly21a,AubPerZde19,GamKizPerXu22,DingSun19,Huang24,Stojnicbinperflrdt23} perceptrons in machine learning/neural nets; statistical regressions and compressed sensing in  information theory and signal processing  \cite{StojnicUpper10,StojnicCSetam09,StojnicGenLasso10,BayMon10,BayMon10lasso,DonohoPol}; and satisfiability \cite{CojOgh14,DinSlySun14,MezParZec02,MerMezZec06,FriWor05,AchPer04} and assignment/matching \cite{AlmSor02,Wast12,Wast09,TalAssignment03,Talbook11a,Talbook11b,NaiPraSha05,CopSor99,Ald01,Ald92,MezPar87} problems in algorithmic computer science and optimization theory. Moreover, the above thinking pattern has sufficiently matured over the years that in many instances it is utilized even if there is no concrete natural phenomena to which underlying statistical mechanics practically relates.

In this paper, we are interested in a slightly different aspect of SK models. Namely, instead of focusing on widening applicability horizons, we focus on core SK model discoveries. As stated earlier, we look at algorithmic properties with a particular emphasis on simple practically realizable implementations. Before we get to relevant technical discussions we revisit some of the key milestones achieved while studying SK models over the last 50 years (an interesting recent perspective of this type by the inventors of the SK model themselves can be found in \cite{SheKir25}).

\section{SK model --- mathematical preliminaries and prior work}
\label{sec:skmodel}

The SK form given in (\ref{eq:inteq1}) comes as a special case of a more general statistical mechanics concept. Namely, one starts with the so-called Hamiltonian
\begin{equation}
\cH(G)=\sum_{i=1}^{n}\sum_{j=1}^{n}  A_{ij}\x_i\x_j,\label{eq:ham}
\end{equation}
where in the SK scenario
\begin{equation}
A_{ij}(G)= G_{ij},\label{eq:hamAij}
\end{equation}
are the so-called quenched interactions. A simple combination  of (\ref{eq:ham}) and (\ref{eq:hamAij}) gives
\begin{equation}
\cH(G)= \x^TG\x.\label{eq:a0ham1}
\end{equation}
Slightly differently from a common literature practice, the diagonal elements are included as well. Such inclusion facilitates writing without significantly affecting any of the analytical considerations.  The  thermodynamic limit  (or, mathematically speaking, $n\rightarrow\infty$ linear regime) is of our prevalent interest. As stated earlier, we consider random $G$ with independent standard normal entries. One then has for partition function
\begin{equation}
Z(\beta,G)=\sum_{\x\in\mB^n}e^{\beta\cH(G)},\label{eq:a0partfun}
\end{equation}
where $\beta>0$ is the so-called  \emph{inverse temperature} parameter. Given that $Z(\beta,G)$ is random one often considers the thermodynamic limit (average) free energy
\begin{equation}
f_{sq}(\beta)=\lim_{n\rightarrow\infty}\frac{\mE_G\log{(Z(\beta,G)})}{\beta \sqrt{2n}}
=\lim_{n\rightarrow\infty} \frac{\mE_G\log{(\sum_{\x\in\mB^n} e^{\beta\cH(G)})}}{\beta \sqrt{2n}},\label{eq:a0logpartfunsqrt}
\end{equation}
with $\mE_G$ denoting the expectation with respect to $G$ (the adopted convention throughout the paper is that the subscript next to $\mE$ denotes the randomness with respect to which the expectation is evaluated). We are now in position to formally recognize the earlier mentioned  ground state free energy. Namely, in addition to $n\rightarrow\infty$, one also allows so-called zero-temperature limit,
$\beta\rightarrow\infty$, and writes
\begin{equation}
f_{sq}(\infty)   \triangleq  \lim_{\beta\rightarrow\infty}f_{sq}(\beta)  =
\lim_{\beta,n\rightarrow\infty}\frac{\mE_G\log{(Z(\beta,G)})}{\beta \sqrt{2n}}=
 \lim_{n\rightarrow\infty}\frac{\mE_G \max_{\x\in\mB^n} \cH(G)}{\sqrt{2n}}  = \lim_{n\rightarrow\infty}\frac{\mE_G \max_{\x\in\mB^n}\x^TG\x}{\sqrt{2n}}.
  \label{eq:a0limlogpartfunsqrt}
\end{equation}

Our main focus are precisely the above ground state regime and the corresponding free energy (\ref{eq:a0limlogpartfunsqrt}). However, we work with the general form of the free energy given in (\ref{eq:a0logpartfunsqrt}) and  eventually deduce the ground state behavior as a special case.

\subsection{Prior work milestones}
\label{sec:skmodelkey}

\noindent \underline{\emph{\textbf{Theoretical limits:}}} After the nearest neighbor (Edwards-Anderson) spin-glass model was proposed in \cite{EdwAnd75}, its a long range SK antipode followed in  \cite{SheKir72}. The replica method of \cite{EdwAnd75} strengthen by the s-called replica symmetry (RS) ansatz was used to analyze both models. \cite{SheKir72}  observed a negative entropy crisis (unallowable in discrete models) which signalled problems with the proposed methodology. At the time the origin of the problem was not known and could be attributed to both the replica methodology as a whole and/or the symmetry ansatz as its a particular component. Moreover, the ensuing numerical simulations of the ground state energy \cite{SheKir78} significantly deviated from the RS prediction. A potential mismatch resolution arrived with  Parisi's breakthrough discovery of the replica symmetry breaking (RSB) scheme \cite{Par79,Par80,Parisi80}. Early calculations showed that RSB  is capable of both neutralizing the negative entropy and correcting the free energy values. The corrections turned out to be more pronounced for larger $\beta$. Consequently, the ground state $\beta\rightarrow\infty$ regime was affected the most with RS prediction $f_{sq}(\infty)\approx 0.7979$ lowered to $f_{sq}(\infty)\approx 0.7636$  by the Parisi RSB  with two steps of breaking. This was also in a much better agreement with $\sim 0.75\pm 0.01$ estimate obtained through an interpolation of the numerical simulations in \cite{SheKir78}. All of this strengthen belief in the soundness of Parisi RSB and already throughout the eighties of the last century it became a golden standard in utilization of replica methods. It allowed studying many other properties and extensions to previously unfathomably wide range of applications.

Despite producing results that were widely believed to be true, mathematically rigorous confirmations of Parisi's predictions lagged behind. About 25 years after the original Parisi invention,  Guerra  \cite{Guerra03} and Talagrand  \cite{Tal06} proved that the RSB characterization of the SK model is indeed correct. Panchenko then developed his own proof \cite{Pan10,Pan10a,Pan13,Pan13a} which was powerful enough to additionally establish the validity of the \emph{ultrametricity} property precisely as predicted earlier by Parisi \cite{Par83}. Many extensions followed as well (see, e.g., \cite{AuffC15,AuffCZ20,JagTob16} and references therein). All of the above effectively settled the SK model theoretical aspects and established validity of all key Parisi's predictions.

\noindent \underline{\emph{\textbf{Algorithmic aspects:}}}  After determining the theoretical limiting value of the ground state energy, the first next question from the optimization and algorithmic point of view is whether or not a spin configuration that achieves such a value can be efficiently  found. As stated at the beginning, within the NP complexity theory approximating the optimum of an indefinite quadratic form in (\ref{eq:inteq1})  within a $\log(n)^{const.}$ factor  is hard  \cite{AroraBKSH05} (one should note that this is in a strike contrast with the positive semi-definite form, where for  $G\geq 0$ semi-definite relaxations ensure approximation within constant factor of $\frac{2}{\pi}$ \cite{nest97}). Due to their reliance on the  worst case principles, the  NP concepts are usually not a very viable option  to properly assess \emph{typical} solvability. Fairly often it is even not easy to find worst case instances and a meticulous effort is needed to carefully tailor them. Moreover, they are usually deterministic and are either isolated examples or belong to  a class of examples of small size (compared to the size of the set of all possible instances).

A more faithful typical representation is achieved via randomness. While the random structure of $G$ is the intrinsic statistical mechanics feature of the SK model, we should point out two particularly relevant aspects of such structures even in scenarios where the models do not impose randomness per se: \textbf{\emph{(i)}} Compared to the worst case instances, imposing randomness usually emulates in a much better fashion the typical behavior (behavior of a large number or even a majority of all problem instances);  \textbf{\emph{(ii)}} In many scenarios the so-called statistical \emph{universality} property holds which means that for almost all well behaved statistics (say those that can be pushed through the central limit theorem) the results ultimately mimic the ones obtained for standard normal distribution, thereby extending even further emulation of typical behavior.

When it comes to algorithmic aspects of random variants of  (\ref{eq:inteq1}) (i.e., the true SK model) much less was known until fairly recently. One first needs to keep in mind that the bar is set much higher in random settings. As noted in \cite{StojnicHopBnds10},  utilizing the leading eigenvector rounding one can get algorithmic ground state free energy lower bound $\frac{2}{\pi} \approx 0.6366$. The spectral methods give $1$ as an upper bound based on the leading singular values (interestingly, SDP relaxations that match the spectral upper bound can not be improved with higher order analogous relaxation hierarchies \cite{MontanariS16,BandeiraKW20}). All of this implies that, differently from worst case instances, the random ones can be approximated within constant factor in polynomial time. One then naturally wonders, how close to 1 the approximating factor can be?  This, on the other hand, is directly related to the existence of the so-called \emph{computational gaps}. In particular, if the approximative factor can indeed come arbitrarily close to 1, then SK model has no computational gap, which would be another example of a problem that is viewed as hard within the classical NP theory but is in fact typically easy (for more on similar phenomena already known to exist in different problems see, e.g., \cite{Gamar21,GamarSud14,GamarSud17,GamarSud17a} and references therein).

Excluding simple spectral observations not much was known even regarding close derivatives of SK model. In \cite{AbM19}, a CREM (continuous random energy model) was considered and shown to be solvable in polynomial time. While somewhat artificial,  for the purpose of studying algorithmic properties CREM is viewed as sufficiently resembling of SK. In \cite{Subag21}\ Subag designed an efficient algorithm to solve $p$-spin spherical model (a very close associate of SK). A big breakthrough regarding SK model itself arrived with the appearance of Montanari's \cite{Montanari19}. Building on modifications of the message passing algorithms \cite{Kaba03,DonMalMon09,DonohoMM11}, \cite{Montanari19} focused on approximate message passing (AMP)  \cite{DonMalMon09} and considered its an incremental variant, IAMP. Motivated by Subag's spherical SK results, Montanari showed that IAMP approaches arbitrarily closely the (binary/Ising) SK ground state energy provided that the functional relation of Parisi RSB parameters is monotonically increasing (or alternatively that the so-called overlap gap property (OGP) \cite{Gamar21,GamarSud14,GamarSud17,GamarSud17a,AchlioptasCR11,HMMZ08,MMZ05} is absent).  These properties are widely believed to be true with numerical evaluations providing an overwhelming confirmation (a formal proof is still missing, but it is more a technical than conceptual matter). We should also add that \cite{AlaouiMS22} extended the results of \cite{Montanari19} so that they encompass the $p$-spin (binary/Ising) SK models as well. This  in a way completed an analogy with Subag's $p$-spin spherical SK (see, also \cite{AlaouiMS22,AlaouiMS25,HuangS22,Subag24,Subag17,Subag17a} for further results in these and closely related directions).

\noindent \underline{\emph{\textbf{Practical aspects:}}}   The above effectively states that (provided the absence of the OGP) the ground state energy of the SK model --  binary maximization of an  indefinite random (Gaussian) quadratic form -- is computable in polynomial time. In fact, theoretical complexity of IAMP is actually quadratic in $n$. To practically run IAMP, one, however, needs a quick access to Parisi parameters. To what degree the accuracy of Parisi RSB evaluations ultimately impacts concrete practical IAMP implementations remains to be seen.

\subsection{Our contributions}
\label{sec:cont}

\emph{Controlled Loosening-up} (CLuP) algorithmic mechanism was introduced in \cite{Stojnicclupint19,Stojnicclupspreg20}  as a powerful alternative to AMP and other existing state of the art methods for solving \emph{planted} models (two classical applications of binary and  sparse  regression in noisy MIMO ML detection and compressed sensing showcased CLuP's practical  power in \cite{Stojnicclupint19,Stojnicclupspreg20}).  Avoiding some of the AMP's potentially unwanted  features (statistical sensitivity, excessive reliance of planted signal's a priori available knowledge and super large dimensions)  while retaining its accuracy, allowed CLuP to become a desirable alternative for many well known planted models. We here propose a simple CLuP like implementation adapted to fit the SK model and show that it is capable of achieving excellent performance in non-planted scenarios as well.

To do so, for a fixed $r_x$ ($0<r_x\leq 1$) we first recognize the importance of the following associated
\begin{eqnarray}\label{eq:clupskeq1}
\hspace{-1.5in} \mbox{\bl{\textbf{\emph{CLuP-SK model:}}}}  \hspace{.5in}  \max_{\x\in\cX(r_x)} \x^TG\x,
\end{eqnarray}
with
\begin{eqnarray}\label{eq:clupskeq1}
 \cX(r_x) \triangleq \left \{ \x | \x\in\mR^n,\|\x\|_2=r_x,\x_i^2\leq \frac{1}{n}  \right \}.
\end{eqnarray}
Defining corresponding Hamiltonian and partition function
\begin{equation}
\cH_{csk}(G)= \sum_{\x\in\cX(r_x)}\x^TG\x,\label{eq:cska0ham1}
\end{equation}
and
\begin{equation}
Z_{csk}(\beta,G)=\sum_{\x\in\cX(r_x)}e^{\beta\cH_{csk}(G)},\label{eq:cska0partfun}
\end{equation}
we have for the thermodynamic limit (average) free energy
\begin{equation}
f_{csk}(\beta)=\lim_{n\rightarrow\infty}\frac{\mE_G\log{(Z_{csk}(\beta,G)})}{\beta \sqrt{2n}}
=\lim_{n\rightarrow\infty} \frac{\mE_G\log{(\sum_{\x\in\cX(r_x)} e^{\beta\cH_{csk}(G)})}}{\beta \sqrt{2n}}.\label{eq:cska0logpartfunsqrt}
\end{equation}
In the  ground state regime one then has
\begin{eqnarray}
\xi(r_x) \triangleq f_{csk}(\infty)  & \triangleq  & \lim_{\beta\rightarrow\infty}f_{csk}(\beta)  =
\lim_{\beta,n\rightarrow\infty}\frac{\mE_G\log{(Z_{csk}(\beta,G)})}{\beta \sqrt{2n}}
\nonumber \\
& = &
 \lim_{n\rightarrow\infty}\frac{\mE_G \max_{\x\in\cX(r_x)} \cH_{csk}(G)}{\sqrt{2n}}  = \lim_{n\rightarrow\infty}\frac{\mE_G \max_{\x\in\cX(r_x)}\x^TG\x}{\sqrt{2n}}.
  \label{eq:cska0limlogpartfunsqrt}
\end{eqnarray}
In Section \ref{sec:randlincons} we connect the above model to recent progress in studying  random processes \cite{Stojnicsflgscompyx23,Stojnicnflgscompyx23}. Utilizing  \emph{fully lifted} random duality theory (fl RDT) and its a \emph{stationarized}  sfl RDT variant  \cite{Stojnicflrdt23},  we characterize the CLuP-SK model. In Section \ref{sec:algimp} we present the associated CLuP-SK algorithmic procedure and show how the results related to  CLuP-SK and an alternative  $\overline{\mbox{CLuP-SK}}$  model directly translate to such a procedure. An excellent agreement between theoretical predictions and algorithmic performance is observed. In particular, for fairly small dimensions on the order of few thousands the CLuP-SK algorithm practically achieves $\sim 0.76$ ground state free energy.

\section{Connecting CLuP-SK model and sfl RDT}
\label{sec:randlincons}

One first observes  that
\begin{eqnarray}
f_{csk}(\beta) & = &\lim_{n\rightarrow\infty} \frac{\mE_G\log\lp \sum_{\x\in\cX(r_x) }e^{\beta\x^TG\x)}\rp}{\beta \sqrt{2n}},\label{eq:hmsfl1}
\end{eqnarray}
is basically a function of  quirp (quadratically indexed random process) $\x^TG\x$. The machinery of  \cite{Stojnicsflgscompyx23,Stojnicnflgscompyx23} is presented for blirps (bilinearly indexed random processes). However, with very minimal modifications it automatically applies to a much wider range of random processes. In, particular, since quirps  are much simpler than blirps, it immediately applies to them as well. To connect the problems that we study here to \cite{Stojnicsflgscompyx23,Stojnicnflgscompyx23}, we need several technical preliminaries. Consider $r\in\mN$, $k\in\{1,2,\dots,r+1\}$, real scalar $r_x>0$, set $\cX(r_x)\subseteq \mR^n$, and function $f_S(\cdot):\mR^n\rightarrow R$. Let vectors $\q=[\q_0,\q_1,\dots,\q_{r+1}]$ and $\c=[\c_0,\c_1,\dots,\c_{r+1}]$ be such that
 \begin{eqnarray}\label{eq:hmsfl2}
1=\q_0\geq \q_1\geq \q_2\geq \dots \geq \q_r\geq \q_{r+1} & = & 0,
 \end{eqnarray}
$\c_0=1$, $\c_{r+1}=0$, and let ${\mathcal U}_k\triangleq [u^{(4,k)},\u^{(2,k)},\h^{(k)}]$  be such that the elements of  $u^{(4,k)}\in\mR$, $\u^{(2,k)}\in\mR^m$, and $\h^{(k)}\in\mR^n$ are independent standard normals. After setting
  \begin{eqnarray}\label{eq:fl4}
\psi_{S,\infty}(f_{S},\calX,\q,\c,r_x)  =
 \mE_{G,{\mathcal U}_{r+1}} \frac{1}{n\c_r} \log
\lp \mE_{{\mathcal U}_{r}} \lp \dots \lp \mE_{{\mathcal U}_2}\lp\lp\mE_{{\mathcal U}_1} \lp \lp Z_{S,\infty}\rp^{\c_2}\rp\rp^{\frac{\c_3}{\c_2}}\rp\rp^{\frac{\c_4}{\c_3}} \dots \rp^{\frac{\c_{r}}{\c_{r-1}}}\rp,
 \end{eqnarray}
with
\begin{eqnarray}\label{eq:fl5}
Z_{S,\infty} & \triangleq & e^{D_{0,S,\infty}} \nonumber \\
 D_{0,S,\infty} & \triangleq  & \max_{\x\in\cX,\|\x\|_2=r_x}
 \lp f_{S}
+\sqrt{n}  r_x    \lp\sum_{k=2}^{r+1}c_k\h^{(k)}\rp^T\x
+ \sqrt{n} r_x \x^T\lp\sum_{k=2}^{r+1}c_k\u^{(2,k)}\rp \rp \nonumber  \\
 c_k & \triangleq & c_k(\q)=\sqrt{\q_{k-1}-\q_k}.
 \end{eqnarray}
we are in position to recall on the following fundamental sfl RDT result.
\begin{theorem} [\cite{Stojnicflrdt23}]
\label{thm:thmsflrdt1}  Assume large  $n$  regime with $\alpha=\lim_{n\rightarrow\infty} \frac{m}{n}$, remaining constant as  $n$ grows. Let the elements of  $G\in\mR^{m\times n}$ be independent standard normals  and for a scalar $r_x>0$ let $\cX(r_x)\subseteq \mR^n$ be a given set. Assume the complete sfl RDT frame from \cite{Stojnicsflgscompyx23} and let  $f(\x):R^n\rightarrow R$ be a given function. Set
\begin{align}\label{eq:thmsflrdt2eq1}
   \psi_{rp}(r_x) & \triangleq - \max_{\x\in\cX(r_x)}  \lp f(\x)+\x^TG\x \rp
   \qquad  \mbox{(\bl{\textbf{random primal}})} \nonumber \\
   \psi_{rd}(\q,\c,r_x) & \triangleq    \frac{r_x^4}{2}    \sum_{k=2}^{r+1}\Bigg(\Bigg.
   \q^2_{k-1}
   -\q^2_{k}
  \Bigg.\Bigg)
\c_k
  - \psi_{S,\infty}(f(\x),\calX(r_x),\q,\c,r_x) \hspace{.24in} \mbox{(\bl{\textbf{fl random dual}})}.
 \end{align}
Let $\hat{\q_0}\rightarrow 1$,$\hat{\c_0}\rightarrow 1$, $\hat{\q}_{r+1}=\hat{\c}_{r+1}=0$, and let the non-fixed parts of $\hat{\q}\triangleq \hat{\q}(r_x)$ and  $\hat{\c}\triangleq \hat{\c}(r_x)$ be the solutions of the following system
\begin{eqnarray}\label{eq:thmsflrdt2eq2}
    \frac{d \psi_{rd}(\q,\c,r_x)}{d\q} =  0,\quad
   \frac{d \psi_{rd}(\q,\c,r_x)}{d\c} =  0.
 \end{eqnarray}
 Then,
\begin{eqnarray}\label{eq:thmsflrdt2eq3}
    \lim_{n\rightarrow\infty} \frac{\mE_G  \psi_{rp}}{\sqrt{n}}
  & = &
 \lim_{n\rightarrow\infty} \psi_{rd}(\hat{\q}(r_x),\hat{\c}(r_x),r_x) \qquad \mbox{(\bl{\textbf{strong sfl random duality}})},
 \end{eqnarray}
where $\psi_{S,\infty}(\cdot)$ is as in (\ref{eq:fl4})-(\ref{eq:fl5}).
 \end{theorem}
\begin{proof}
Follows after repeating line-by-line derivations in \cite{Stojnicflrdt23,Stojnicnflgscompyx23,Stojnicsflgscompyx23} with cosmetic change $x\rightarrow r_x$ and trivial $\y\rightarrow \x$, $y\rightarrow x$, and  $b_k\rightarrow c_k$ symmetry adjustments.
 \end{proof}

As mentioned in \cite{Stojnicflrdt23}, various  probabilistic variants of (\ref{eq:thmsflrdt2eq3})  hold immediately as well. We skip stating these trivialities and instead focus on practical utilization.

\subsection{Practical utilizations}
\label{sec:prac}

Handling specialization of the above theorem obtained for $f(\x)=0$ and
\begin{equation}\label{eq:prac0a0}
\cX(r_x)=\left \{ \x| \|\x\|_2=r_x, \x_i^2\leq \frac{1}{n}\right \},
\end{equation}
will be the key for everything that follows. Along the same lines, unless otherwise stated, throughout the rest of the paper $\cX$ is assumed to have precisely the form given in (\ref{eq:prac0a0}). One then starts by recognizing that the so-called \emph{random dual} is the key object of practical interest
\begin{align}\label{eq:prac1}
    \psi_{rd}(\q,\c,r_x) & \triangleq    \frac{1}{2}    \sum_{k=2}^{r+1}\Bigg(\Bigg.
   \q^2_{k-1}
   -\q^2_{k}
  \Bigg.\Bigg)
\c_k
  - \psi_{S,\infty}(0,\calX(r_x)\q,\c,r_x). \nonumber \\
  & =   \frac{1}{2}    \sum_{k=2}^{r+1}\Bigg(\Bigg.
   \q^2_{k-1}
   -\q^2_{k}
  \Bigg.\Bigg)
\c_k
  - \frac{1}{n}\varphi(D^{(bin)}(r_x)),
  \end{align}
where similarly to (\ref{eq:fl4})-(\ref{eq:fl5})
  \begin{eqnarray}\label{eq:prac2}
\varphi(D,\c) & = &
 \mE_{G,{\mathcal U}_{r+1}} \frac{1}{\c_r} \log
\lp \mE_{{\mathcal U}_{r}} \lp \dots \lp \mE_{{\mathcal U}_3}\lp\lp\mE_{{\mathcal U}_2} \lp
\lp    e^{D}   \rp^{\c_2}\rp\rp^{\frac{\c_3}{\c_2}}\rp\rp^{\frac{\c_4}{\c_3}} \dots \rp^{\frac{\c_{r}}{\c_{r-1}}}\rp,
  \end{eqnarray}
and
\begin{eqnarray}\label{eq:prac3}
D^{(bin)}(r_x) & = & \max_{\x\in\cX(r_x)} \lp  r_x \sqrt{2n}      \lp\sum_{k=2}^{r+1}c_k\h^{(k)}\rp^T\x  \rp.
 \end{eqnarray}
We then find
\begin{eqnarray}\label{eq:prac4}
D^{(bin)}(r_x) & = & \max_{\x\in\cX(r_x)} \lp r_x  \sqrt{2n}      \lp\sum_{k=2}^{r+1}c_k\h^{(k)}\rp^T\x  \rp
 = - r_x \sum_{i=1}^n D^{(bin)}_i(c_k) +\gamma r_x^3,
 \end{eqnarray}
where
\begin{eqnarray}\label{eq:prac5}
D^{(bin)}_i(c_k)= \begin{cases}
                    -\frac{\lp\sum_{k=2}^{r+1}c_k\h_i^{(k)}\rp ^2}{2\gamma}, & \mbox{if } \left |  \sum_{k=2}^{r+1} c_k \h_i^{(k)}   \right | \leq \sqrt{2}\gamma  \\
                    -\sqrt{2} \left | \sum_{k=2}^{r+1} c_k\h_i^{(k)}   \right | +\gamma, & \mbox{otherwise}.
                  \end{cases}
\end{eqnarray}
Connecting  $f_{csk}$ from (\ref{eq:cska0limlogpartfunsqrt}) and the random primal $\psi_{rp}(r_x)$ from Theorem \ref{thm:thmsflrdt1}, we observe
 \begin{eqnarray}
\xi(r_x) = f_{csk}(\infty)
& = &  \lim_{n\rightarrow\infty}\frac{\mE_G \max_{\x\in\cX(r_x)} \x^TG\x  }{\sqrt{2n}}
=
      -\lim_{n\rightarrow\infty} \frac{\mE_G  \psi_{rp}}{\sqrt{2n}}
   =
 -\frac{1}{\sqrt{2}}\lim_{n\rightarrow\infty} \psi_{rd}(\hat{\q},\hat{\c},r_x),
  \label{eq:prac11}
\end{eqnarray}
with the non-fixed parts of $\hat{\q}$ and  $\hat{\c}$ being the solutions of
\begin{eqnarray}\label{eq:prac12}
    \frac{d \psi_{rd}(\q,\c,r_x)}{d\q} =  0,\quad
   \frac{d \psi_{rd}(\q,\c,r_x)}{d\c} =  0.
 \end{eqnarray}
 Relying on (\ref{eq:prac1})-(\ref{eq:prac5}), we further have
 \begin{eqnarray}
 \lim_{n\rightarrow\infty} \psi_{rd}(\hat{\q},\hat{\c},r_x) =  \bar{\psi}_{rd}(\hat{\q},\hat{\c},\hat{\gamma}_{sq},r_x),
  \label{eq:prac12a}
\end{eqnarray}
where
\begin{eqnarray}\label{eq:prac13}
    \bar{\psi}_{rd}(\q,\c,\gamma_{sq},r_x)   & = &  \frac{1}{2}    \sum_{k=2}^{r+1}\Bigg(\Bigg.
   \q^2_{k-1}
   - \q^2_{k}
  \Bigg.\Bigg)
\c_k
-\gamma r_x^3
- \varphi(D_1^{(bin)}(c_k(\q)),\c).
  \end{eqnarray}
In (\ref{eq:prac13}), $\varphi(D_1^{(bin)}(c_k(\q)),\c)$  is   (based on (\ref{eq:prac2}) and (\ref{eq:prac5})) given by
\begin{align}\label{eq:prac14}
\varphi(D_1^{(bin)}(c_k(\q)),\c) & =
 \mE_{{\mathcal U}_{r+1}} \frac{1}{\c_r} \log
\lp \mE_{{\mathcal U}_{r}} \lp \dots \lp \mE_{{\mathcal U}_3}\lp\lp\mE_{{\mathcal U}_2} \lp
    e^{  -r_x \c_2 D_1^{(bin)}(c_k(\q))  }  \rp\rp^{\frac{\c_3}{\c_2}}\rp\rp^{\frac{\c_4}{\c_3}} \dots \rp^{\frac{\c_{r}}{\c_{r-1}}}\rp, \nonumber \\
  \end{align}
whereas $\hat{\gamma}$ and  the non-fixed parts of $\hat{\q}$, and  $\hat{\c}$ are the solutions of
\begin{eqnarray}\label{eq:prac16}
    \frac{d \bar{\psi}_{rd}(\q,\c,\gamma_{sq},r_x)}{d\q} & = & 0 \nonumber \\
   \frac{d \bar{\psi}_{rd}(\q,\c,\gamma_{sq},r_x)}{d\c} & = & 0 \nonumber \\
   \frac{d \bar{\psi}_{rd}(\q,\c,\gamma_{sq},r_x)}{d\gamma} & = & 0.
 \end{eqnarray}
One then easily observes
\begin{eqnarray}\label{eq:prac17}
c_k(\hat{\q})  & = & \sqrt{\hat{\q}_{k-1}-\hat{\q}_k},
 \end{eqnarray}
and after connecting  (\ref{eq:prac11}) to (\ref{eq:prac12a}) finds
 \begin{eqnarray}
f_{csk}(\infty)
& = &  \lim_{n\rightarrow\infty}\frac{\mE_G \max_{\x\in\cX(r_x)} \x^TG\x}{\sqrt{2n}}  =
 -\frac{1}{\sqrt{2}}\lim_{n\rightarrow\infty} \psi_{rd}(\hat{\q},\hat{\c},r_x)
 = - \frac{1}{\sqrt{2}} \bar{\psi}_{rd}(\hat{\q},\hat{\c},\hat{\gamma},r_x) \nonumber \\
 & = & \frac{1}{\sqrt{2}}\lp  -\frac{1}{2}    \sum_{k=2}^{r+1}\Bigg(\Bigg.
    \hat{\q}^2_{k-1}
   - \hat{\q}^2_{k}
  \Bigg.\Bigg)
\hat{\c}_k
 + \hat{\gamma} r_x^3
  + \varphi(D_1^{(bin)}(c_k(\hat{\q})),\c) \rp. \nonumber \\
  \label{eq:prac18}
\end{eqnarray}
The above results are summarized in the following theorem.

\begin{theorem}
  \label{thme:thmprac1}
Consider linear large $n$  regime with $\alpha=\lim_{n\rightarrow\infty} \frac{m}{n}$ and  assume the complete sfl RDT setup of \cite{Stojnicsflgscompyx23}.  Let $\varphi(\cdot)$ and $\bar{\psi}_{rd}(\cdot)$ be as given in (\ref{eq:prac2}) and (\ref{eq:prac13}), respectively and let the ``fixed'' parts of $\hat{\q}$, and $\hat{\c}$ be $\hat{\q}_1\rightarrow 1$, $\hat{\c}_1\rightarrow 1$, $\hat{\q}_{r+1}=\hat{\c}_{r+1}=0$. Also, let $\hat{\gamma}$ and the ``non-fixed'' parts of $\hat{\q}_k$, and $\hat{\c}_k$ ($k\in\{2,3,\dots,r\}$) be the solutions of (\ref{eq:prac16}). For $c_k(\hat{\q})$  from (\ref{eq:prac17}), one has
 \begin{eqnarray}
\xi(r_x) \triangleq f_{csk}(\infty)
& = &   \frac{1}{\sqrt{2}} \lp -\frac{1}{2}    \sum_{k=2}^{r+1}\Bigg(\Bigg.
    \hat{\q}^2_{k-1}
   - \hat{\q}^2_{k}
  \Bigg.\Bigg)
\hat{\c}_k
+ \hat{\gamma} r_x^3
  + \varphi(D_1^{(bin)}(c_k(\hat{\q})),\hat{\c})   \rp. \nonumber \\
  \label{eq:thmprac1eq1}
\end{eqnarray}
\end{theorem}
\begin{proof}
Follows automatically from the above discussion, Theorem \ref{thm:thmsflrdt1}, and the sfl RDT machinery presented in \cite{Stojnicnflgscompyx23,Stojnicsflgscompyx23,Stojnicflrdt23}.
\end{proof}

\subsection{Numerical evaluations}
\label{sec:nuemrical}

Theorem \ref{thme:thmprac1} reaches its full practically relevance if all underlying numerical evaluations cam be conducted. We show next that this can indeed be done. To allow for a systematic exposition, we start with first level of lifting, i.e., with $r=1$.

\underline{1) \textbf{\emph{$r=1$ -- first level of lifting:}}} For $r=1$ we have $\hat{\q}_1\rightarrow 1$ which together with $\hat{\q}_{r+1}=\hat{\q}_{2}=0$, and $\hat{\c}_{2}\rightarrow 0$ allows to write
\begin{align}\label{eq:prac19}
    -\bar{\psi}_{rd}^{(1)} (\hat{\q},\hat{\c},\gamma,r_x)   & =   -\frac{1}{2}
\c_2 r_x^4
+\gamma r_x^3
  + \frac{1}{\c_2}\log\lp \mE_{{\mathcal U}_2} e^{  -r_x \c_2 D_1^{(bin)}(c_k(\q))   }\rp
 \nonumber \\
& \rightarrow
    \gamma r_x^3
+   \frac{1}{\c_2}\log\lp 1 -r_x \c_2 \mE_{{\mathcal U}_2}  D_1^{(bin)}(c_k(\q))  \rp
 \nonumber \\
& \rightarrow
    \gamma r_x^3
-   r_x  \mE_{{\mathcal U}_2}  D_1^{(bin)}(c_k(\q))
  \nonumber \\
& \rightarrow
    \gamma r_x^3
+  f_q^{(1)},
  \end{align}
  where
  \begin{eqnarray}\label{eq:prac19a0}
  f_q^{(1)} = -r_x \mE_{{\mathcal U}_2}  D_1^{(bin)}(c_k(\q)).
  \end{eqnarray}
After setting
  \begin{eqnarray}\label{eq:prac19a1}
f_{q,1}^{(1)} & = &  2 \lp - \sqrt{2} / \sqrt{2\pi} e^{-\gamma^2} + \frac{\gamma}{2} \erfc(\gamma) \rp  \nonumber \\
f_{q,2}^{(1)} & = &  2/2/\gamma  (  \sqrt{2}\gamma  e^{-\gamma^2}/\sqrt{2\pi} - (1/2 - 1/2 \erfc(\gamma) )  ),
   \end{eqnarray}
and solving the remaining integrals one finds
  \begin{eqnarray}\label{eq:prac19a2}
f_q^{(1)} = -r_x \lp f_{q,1}^{(1)} + f_{q,2}^{(1)} \rp.
  \end{eqnarray}
One then ealso has
\begin{align}\label{eq:prac20}
  f^{(1)}_{csk}(\infty)=-\frac{1}{\sqrt{2}} \bar{\psi}_{rd}^{(1)}(\hat{\q},\hat{\c},\hat{\gamma},r_x).
  \end{align}
For example, for $r_x=1$ one has $\hat{\gamma}\rightarrow 0$, $f_{q,1}^{(1)}     \rightarrow -\frac{2}{\sqrt{\pi}} $, $f_{q,2}^{(1)}  \rightarrow 0$ and
\begin{align}\label{eq:prac21}
  f^{(1)}_{csk}(\infty) \rightarrow  - \frac{1}{\sqrt{2}} \lp f_{q,1}^{(1)} + f_{q,2}^{(1)} \rp
\rightarrow   \sqrt{\frac{2}{\pi}} \approx  \bl{\mathbf{0.7979}}.
  \end{align}

\underline{2) \textbf{\emph{$r=2$ -- second level of lifting:}}} For  $r=2$, we have $\hat{\q}_1\rightarrow 1$, $\hat{\q}_{r+1}=\hat{\q}_{3}=0$, but in general  $\hat{\c}_{2}\neq 0$ and $\q_2\neq0$. Similarly to what we did above, we now write
\begin{align}\label{eq:prac24}
   - \bar{\psi}_{rd}^{(2)} (\q,\c,\gamma,r_x)   & =  - \frac{1}{2}
(1-\q^2_2)\c_2r_x^4
   + \gamma r_x^3
  + \frac{1}{\c_2}\mE_{{\mathcal U}_3}\log\lp \mE_{{\mathcal U}_2} e^{   -r_x \c_2 D_1^{(bin)}(c_k(\q))    }   \rp \nonumber \\
& =  - \frac{1}{2}
(1-\q^2_2)\c_2r_x^4
   + \gamma r_x^3
  + f_q^{(2)},
   \end{align}
with
\begin{eqnarray} \label{eq:prac24a0}
f_q^{(2)} = \frac{1}{\c_2}\mE_{{\mathcal U}_3}\log\lp \mE_{{\mathcal U}_2} e^{   -r_x \c_2 D_1^{(bin)}(c_k(\q))    }   \rp.
 \end{eqnarray}
After setting
\begin{eqnarray} \label{eq:prac24a1}
 D & = &  -\sqrt{\q_2} \h_i^{(3)} /\sqrt{1-\q_2} + \sqrt{2}\gamma/\sqrt{1-\q_2} \nonumber \\
C & = & -\sqrt{\q_2} \h_i^{(3)} /\sqrt{1-\q_2} - \sqrt{2}\gamma/\sqrt{1-\q_2} \nonumber \\
E & = &  \frac{ r_x\c_2} {2\gamma }  \q_2  \lp \h_i^{(3)}\rp^2  \nonumber \\
A & = & -  \frac{ r_x\c_2} {2\gamma }     2\sqrt{1-\q_2} \sqrt{\q_2}  \h_i^{(3)}   \nonumber \\
B & = & - \frac{ r_x\c_2} {2\gamma }    ( 1-\q_2  )  + 1/2  \nonumber \\
 f_{q,1}^{(2)} & = & \frac{  e^{ E+A^2/(4 B)  }     (\erf((A + 2 B D)/(2 \sqrt{|B| }    ) ) - \erf((A + 2B C)/(2 \sqrt{ |B|  }  )   ))}
 {
 2 \sqrt{2}
 \sqrt{ |B |  }   },
 \end{eqnarray}
and
\begin{eqnarray} \label{eq:prac24a2}
 C_1 & = & -\sqrt{\q_2} \h_i^{(3)} /\sqrt{1-\q_2} + \sqrt{2}\gamma/\sqrt{1-\q_2} \nonumber \\
 A_1 & = &   r_x\c_2     \sqrt{2}\sqrt{1-\q_2} \nonumber \\
B_1 & = &   r_x\c_2    \sqrt{2} \sqrt{\q_2}  \h_i^{(3)}    \nonumber \\
 f_{q,21}^{(2)} & = & \frac{1}{2} e^{ B_1 + A_1^2/2}   (\erf((A_1 - C_1)/\sqrt{2}  ) + 1)  \nonumber \\
  C_2 & = & -\sqrt{\q_2} \h_i^{(3)} /\sqrt{1-\q_2} - \sqrt{2}\gamma/\sqrt{1-\q_2} \nonumber \\
 A_2 & = &  - r_x\c_2     \sqrt{2}\sqrt{1-\q_2} \nonumber \\
B_2 & = &  - r_x\c_2    \sqrt{2} \sqrt{\q_2}  \h_i^{(3)}    \nonumber \\
 f_{q,22}^{(2)} & = & \frac{1}{2} e^{ B_2 + A_2^2/2}   (\erfc((A_2 - C_2)/\sqrt{2}  ) )  \nonumber \\
 f_{q,2}^{(2)} & = &  e^{-\c_2 r_x\gamma}   \lp f_{q,21}^{(2)} + f_{q,22}^{(2)}\rp,
 \end{eqnarray}
solving the remaining integrals gives
\begin{eqnarray} \label{eq:prac24a3}
f_q^{(2)} = \frac{1}{\c_2}\mE_{{\mathcal U}_3}\log\lp    f_{q,1}^{(2)}  +  f_{q,2}^{(2)}    \rp.
 \end{eqnarray}
For a concrete example  $r_x=1$  considered on the first level, one finds after computing all the derivatives $\hat{\gamma}\rightarrow 0$, $\q_2=0.4768$, $\hat{\c}_2=0.9623$, and
\begin{align}\label{eq:prac25}
\hspace{-2in}(\mbox{\emph{full} second level:}) \qquad \qquad  f^{(2)}_{csk}(\infty) \rightarrow \bl{\mathbf{0.7653}}.
  \end{align}

\underline{3) \textbf{\emph{$r=3$ -- third level of lifting:}}} To make numerical evaluations less cumbersome we will mostly focus on the so-called third \emph{partial} level of lifting (as we will soon see these results will be almost indistinguishable from the second full level). We have $r=3$, $\hat{\q}_1\rightarrow 1$, and  $\hat{\q}_{r+1}=\hat{\q}_{4}=0$. In addition to having  $\hat{\c}_{2}\neq 0$ and $\q_2\neq0$, we now also have $\hat{\c}_{3}\neq 0$ and $\q_3=0$. Analogously to (\ref{eq:prac24}) and (\ref{eq:prac24a0}), we write
\begin{align}\label{eq:prac26}
   - \bar{\psi}_{rd}^{(3,p)} (\q,\c,\gamma,r_x)   & =  - \frac{1}{2}
(1-\q^2_2)\c_2r_x^4 - \frac{1}{2}
\q^2_2\c_3r_x^4
   + \gamma r_x^3
  + \frac{1}{\c_3} \log\lp   \mE_{{\mathcal U}_3}   \lp   \mE_{{\mathcal U}_2} e^{   -r_x \c_2 D_1^{(bin)}(c_k(\q))    }  \rp^{\frac{\c_3}{\c_2}}   \rp \nonumber \\
& =  - \frac{1}{2}
(1-\q^2_2)\c_2r_x^4
 - \frac{1}{2}
\q^2_2\c_3r_x^4
   + \gamma r_x^3
  + f_q^{(3)},
   \end{align}
where
\begin{eqnarray} \label{eq:prac26a0}
f_q^{(3)} = \frac{1}{\c_3} \log\lp   \mE_{{\mathcal U}_3}   \lp   \mE_{{\mathcal U}_2} e^{   -r_x \c_2 D_1^{(bin)}(c_k(\q))    }  \rp^{\frac{\c_3}{\c_2}}   \rp.
 \end{eqnarray}
  Taking again $r_x=1$ as a concrete example and computing all the derivatives one finds $\hat{\gamma}\rightarrow 0$, $\hat{\q}_2= 0.7434$, $\hat{\c}_2= 1.4586$,  $\hat{\c}_3 =  0.3569$, and
\begin{align}\label{eq:prac27}
\hspace{-2in}(\mbox{\emph{partial} third level:}) \qquad \qquad f^{(3,p)}_{csk}(\infty) \rightarrow \bl{\mathbf{0.7640}}.
  \end{align}
Table \ref{tab:2rsbunifiedsqrtpos} summarizes the above results in a systematic way. It shows how all the relevant quantities change as the lifting mechanism progresses through the first three levels.

\begin{table}[h]
\caption{$r$-sfl RDT parameters; CLUP SK model; $r_x=1$;  $\hat{\c}_1\rightarrow 1$; $n,\beta\rightarrow\infty$}\vspace{.1in}
\centering
\def\arraystretch{1.2}
\begin{tabular}{||l||c||c|c|c||c|c||c||}\hline\hline
 \hspace{-0in}$r$-sfl RDT                                             & $\hat{\gamma}$      & $\hat{\q}_3$ & $\hat{\q}_2$  & $\hat{\q}_1$ &  $\hat{\c}_3$   & $\hat{\c}_2$    & $f_{csk}^{(r)} (\infty) $  \\ \hline\hline
$\mathbf{1}$ (full)                                      & $0$ & $0$ & $0$ & $\rightarrow 1$ &  $\rightarrow 0$
 &  $\rightarrow 0$  & \bl{$\mathbf{0.7979}$} \\ \hline\hline
   $\mathbf{2}$ (full)                                      & $0$ & $0$ & $0.4768$ & $\rightarrow 1$ &  $\rightarrow 0$
 &  $0.9623$   & \bl{$\mathbf{0.7653}$}  \\ \hline\hline
 $\mathbf{3}$ (partial)                                      & $0$ & $0$ & $0.7434$ & $\rightarrow 1$ &  $0.3569$
 &  $1.4586$   & \bl{$\mathbf{0.7640}$}   \\ \hline\hline
\end{tabular}
\label{tab:2rsbunifiedsqrtpos}
\end{table}

The above values are given for the concreteness  and are obtained for a particular value $r_x=1$ (as such they also correspond to the plain SK model).  In Figure \ref{fig:fig1} these results are complemented with the results  for all $r_x\in (0,1]$ which correspond to the CLuP-SK model. Two things should be noted: \textbf{\emph{(i)}} The effect of lifting becomes more visible as $r_x$ increases and is maximal for $r_x\rightarrow 1$; and \textbf{\emph{(ii)}} Already on the first level the curve is monotonically increasing and such phenomenology remains intact as one moves to higher levels of lifting (on higher levels the effect becomes less pronounced and the curves are visually almost indistinguishable from the one on the second level  rendering plotting them practically pointless).  As we will see below, monotonicity and absence of local optima play a key role in understanding underlying SK algorithms.
\begin{figure}[h]
\centering
\centerline{\includegraphics[width=1.00\linewidth]{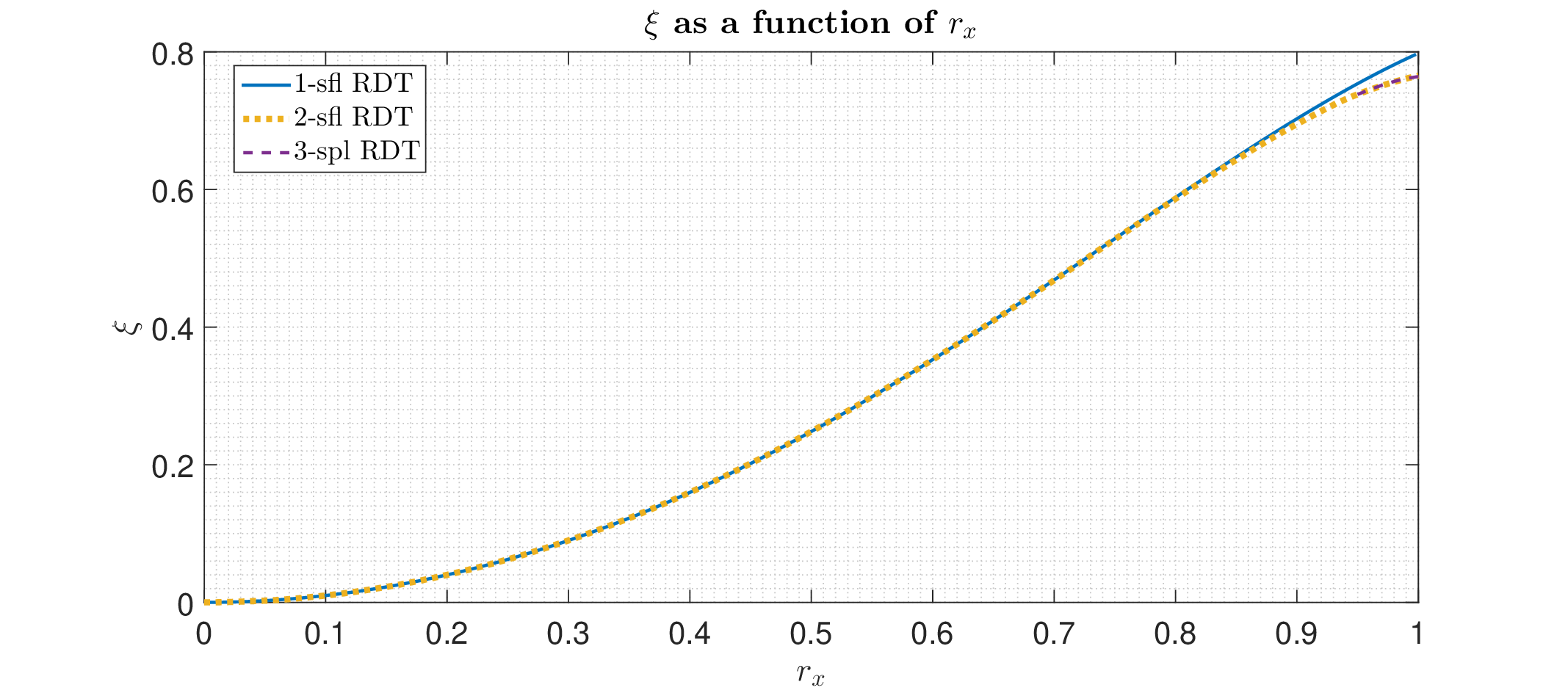}}
\caption{$\xi$ as a function of $r_x$}
\label{fig:fig1}
\end{figure}

\underline{4) \textbf{\emph{$r$-th level of lifting:}}}  For the completeness, we also give the explicit form of $\bar{\psi}_{rd}^{(r)}(\q,\c,\gamma,r_x) $ for general $r$. Mimicking the above procedures, one has
 \begin{align}\label{eq:prac29}
   -  \bar{\psi}_{rd}^{(r)}(\q,\c,\gamma,r_x)   & =  - \frac{1}{2}
\sum_{k=2}^{r+1}(\q^2_{k-1}-\q^2_k)\c_k  +\gamma r_x^3   + f_q^{(r)},    \end{align}
where
\begin{eqnarray}\label{eq:prac30}
f_q^{(r)} = \frac{1}{\c_r}\mE_{{\mathcal U}_{r+1}}\log\lp \mE_{{\mathcal U}_r} \lp \dots \lp \mE_{{\mathcal U}_3} \lp  \mE_{{\mathcal U}_2}
e^{   -r_x \c_2 D_1^{(bin)}(c_k(\q))    }
 \rp^{\frac{\c_3}{\c_2}} \rp^{\frac{\c_4}{\c_{3}}} \dots \rp^{\frac{\c_r}{\c_{r-1}}} \rp.
 \end{eqnarray}
In addition to making the role of all key parameters $\q$, $\c$, and $\gamma$ clearly visible, the above expression also enables evaluation of all derivatives needed at any level of lifting $r\in\mN$.  Moreover, we conducted all of the above numerical evaluations relying on modulo-$\m$ sfl results from \cite{Stojnicsflgscompyx23,Stojnicnflgscompyx23,Stojnicflrdt23} and obtained exactly the same results as in Table \ref{tab:2rsbunifiedsqrtpos} and Figure \ref{fig:fig1}. This basically indicates the  \emph{minimization} type of $\c$ \emph{stationarity} and is in a remarkable agreement with the corresponding discoveries from \cite{Stojnichopflrdt23,Stojnicbinperflrdt23,Stojnicnegsphflrdt23,Stojnicabple25}.

\section{Algorithmic implications}
\label{sec:algimp}

Over the last several decades \emph{message-passing} (MP) algorithms have been among the most successful algorithmic tools for handling hard optimization problems. A specific form called \emph{Approximate message passing} (AMP)
\cite{DonohoMM11,DonMalMon09} has been particularly successful in handling so-called \emph{planted} (or \emph{student-teacher} (ST)) models. Being both sufficiently complex to closely match excellent MP performance and sufficiently simple to allow for strong theoretical support \cite{Bolt14,BayMon10lasso,BayMon10,BayMonconf10},  AMPs  quickly became an irreplaceable algorithmic tool in a host of different scientific fields -- far exceeding their original introduction within compressed sensing context. While AMPs are generically excellent algorithms, some of their features including statistical sensitivity, frequent reliance on planted signal's a priori available knowledge and excessively large underlying dimensions in certain applications may be perceived as obstacles towards universal practicality.

\cite{Stojnicclupint19,Stojnicclupspreg20}   introduced \emph{Controlled Loosening-up} (CLuP) algorithmic mechanism  as a powerful alternative to AMP and other existing state of the art methods for solving planted models (CLuP's excellent performance was showcased via two concrete classical applications of binary \cite{Stojnicclupint19} and  sparse  \cite{Stojnicclupspreg20} regression in  MIMO ML detection and compressed sensing).  As CLuP managed to circumvent some of the key AMP's features while retaining its accuracy it positioned itself as a viable choice for many well known planted models. We here show that a CLuP like implementation is capable of producing excellent results in non-planted scenarios as well. As discussed earlier, we consider the so-called ground state energy of the celebrated Sherrington-Kirkpatrick (SK) model \cite{SheKir72}. In classical optimization theory the problem corresponds to famous indefinite quadratic form maximization over the vertices of the binary cube and is among foundational problems of the NP theory.

\subsection{CLuP-SK algorithmic implementation}
\label{sec:clupskalg}

For an indefinite matrix $G\in\mR^{n\times n}$ we are interested in algorithmic solving of
\begin{eqnarray}\label{eq:algimpeq1}
 \max_{\x} & & \x^TG\x \nonumber\\
 \mbox{subject to} & &  \x_i^2=\frac{1}{n}, 1\leq i\leq n.
 \end{eqnarray}
 We propose the following (iterative)  CLuP-like procedure that we call
\begin{eqnarray}\label{eq:algimpeq2}
 \hspace{-.55in} \mbox{\bl{\textbf{\emph{CLuP-SK algorithm:}}}}  \hspace{.65in} \x^{(t+1)} & \rightarrow  &
\mbox{\textbf{gradbar}}\lp\bar{f}_{b,x} \lp \x;\bar{t}_{0x}^{(t)} \rp ;\x^{(t)},\bar{t}_{0x}^{(t)} \rp
 \nonumber \\
\bar{t}_{0x}^{(t+1)}  &  \rightarrow  & \bar{c}^{(t)}\bar{t}_{0x}^{(t)}.
\end{eqnarray}
Procedure $\mbox{\textbf{gradbar}}\lp\bar{f}_{b,x} \lp \x;\bar{t}_{0x}^{(t)} \rp ;\x^{(t)},\bar{t}_{0x}^{(t)} \rp$ applies standard gradient descent starting from $\x^{(t)}$ to  function $\bar{f}_{b,x} \lp \x;\bar{t}_{0x}^{(t)}  \rp $ specified by an argument $\bar{t}_{0x}^{(t)}$
\begin{eqnarray}\label{eq:algimpeq3}
\bar{f}_{b,x} \lp\x;\bar{t}_{0x}^{(t)}\rp = - \bar{t}_{0x}^{(t)} \|\x\|_2 - \log\lp - \lp \x^T  \lp 0.9 I  -  \frac{1}{2\sqrt{2n}} \lp G^T+G\rp   \rp \x - \kappa \rp \rp
-\frac{1}{n}\sum_{i=1}^{n}  \log(1-n\x_i^2).
\end{eqnarray}
While $\kappa$ is a free parameter, we found the above procedure as not overly sensitive with respect to $\kappa$. Choosing $\kappa=0.155$ in all our numerical experiments sufficed. Other parameters, starting $t_{0x}^{(0)}$ and incremental $c^{(t)}$ are also fairly flexible. For example, $t_{0x}^{(0)}=0.0005$ and  $c^{(t)}=1.1$ are a solid starting choice that can be changed as $n$ varies. $\x^{(0)}$ is basically any $\x$ that fits under $\log$s. One can start with a random choice of $\pm\frac{1}{\sqrt{n}}$ and then scale down by two until a feasible option is reached.

In Table \ref{tab:tab2}  we include the obtained results for $\xi(1)$ (the thermodynamic limit of the ground state energy of the SK model)
\begin{eqnarray}\label{eq:algimpeq4}
\xi(1) = \lim_{n\rightarrow\infty} \frac{\mE_G \max_{\x\in\mB^n}\x^TG\x}{\sqrt{2n}}.
\end{eqnarray}
As the results from the table show one approaches  $\sim 0.76$ limit even for fairly small $n$ on the order of a few thousands. We should also add that the results in the table are obtained without restarting or any other advanced modifications (like say, stochastic descent). In other words, they are obtained for one choice of $\x^{(0)}$ and with plain gradient descent. While the effect of restarts is likely to fade away as  $n\rightarrow \infty$, for finite $n$ (those that can be simulated) adding restarts in general can be  beneficial. As the obtained results are already beyond the best of the expectations, we skipped adding further modifications.

\begin{table}[h]
\caption{Performance of CLuP-SK algorithm; \textbf{\bl{simulated}/theory} }\vspace{.1in}
\centering
\def\arraystretch{1.2}
\begin{tabular}{||l||c||c||c||c|| }\hline\hline
 \hspace{-0in}$n$                                             & $2000$    & $4000$ & $8000$ &  $\infty$ (\textbf{theory}) \\ \hline\hline
$\xi(1)$                                         & \bl{$\mathbf{0.755}$}  & \bl{$\mathbf{0.757}$}  & \bl{$\mathbf{0.758}$} & $\mathbf{0.763}$  \\ \hline\hline
\end{tabular}
\label{tab:tab2}
\end{table}

\subsection{Loss landscape}
\label{sec:algimplossland}

Ideally, one would like that the output of $\mbox{\textbf{gradbar}}\lp\bar{f}_{b,x}\lp\x;\bar{t}_{0x}^{(t)}\rp ;\x^{(t)},\bar{t}_{0x}^{(t)} \rp$ is $ \mbox{argmin}_{\x}  \bar{f}_{b,x} \lp\x;\bar{t}_{0x}^{(t)} \rp $, i.e., one would like
\begin{eqnarray}\label{eq:algimpeq5}
 \x^{(t+1)}\rightarrow  \mbox{argmin}_{\x}  \bar{f}_{b,x} \lp \x;\bar{t}_{0x}^{(t)} \rp.
\end{eqnarray}
Since we are running a descent type of algorithm, the shape (landscape) of the objective (loss) is critically important in achieving global optima. In case of CLuP-SK procedure, the shape of $\bar{f}_{b,x} \lp\x;\bar{t}_{0x}^{(t)}\rp$ is of interest.

\subsubsection{Closed form based \emph{approximate} landscape characterization}
\label{sec:algimplosslandappox}

While simple modifications of the machinery from previous sections is in principle sufficient to analyze $\bar{f}_{b,x} \lp\x;\bar{t}_{0x}^{(t)}\rp$, ensuing numerical evaluations are a bit more involved and possibly more prone to residual numerical instabilities. To circumvent that we analyze closely related function
\begin{eqnarray}\label{eq:algimpeq6}
f_{b,x} \lp\x; t_{0x} \rp = - t_{0x} \|\x\|_2 - \log\lp - \lp \x^T  \lp 0.9 I  -  \frac{1}{2\sqrt{2n}} \lp G^T+G\rp   \rp \x - \kappa \rp \rp.
\end{eqnarray}
In particular, for a fixed $t_{0x}$ we are interested in behavior of
\begin{eqnarray}\label{eq:algimpeq7}
f_{b} \lp r_x \rp = \min_{\x\in\cX(r_x)} - t_{0x} \|\x\|_2 - \log\lp - \lp \x^T  \lp 0.9 I  -  \frac{1}{2\sqrt{2n}} \lp G^T+G\rp   \rp \x - \kappa \rp \rp.
\end{eqnarray}
Recalling on (\ref{eq:prac0a0})
\begin{eqnarray}\label{eq:algimpeq8}
f_{b} \lp r_x \rp
& = & \min_{\x\in\cX(r_x)} - t_{0x}r_x - \log \lp  -0.9r_x^2 + \xi(r_x) + \kappa \rp
=
 - t_{0x}r_x - \log \lp  -0.9r_x^2 + f_{csk}(\infty)  + \kappa \rp .
\end{eqnarray}
In Figures \ref{fig:fig2}-\ref{fig:fig4}, we show $\frac{f_{b} \lp r_x \rp }{t_{0x}}$ for three different value of $t_{0x}$. Results for the first level of full lifting are  obtained  for $\xi(r_x)=f^{(1)}_{csk} (\infty) $; corresponding second level results are  obtained  for $\xi(r_x)=f^{(2)}_{csk} (\infty) $. Additionally, let the optimal $r_x$ be
\begin{eqnarray}\label{eq:algimpeq9}
\hat{r}_x =  \mbox{argmin}_{r_x\in(0,1]} f_{b} \lp r_x \rp.
\end{eqnarray}
We observe the same phenomenological behavior for all three $t_{0x}$ choices. For both first and second lifting level, $\frac{f_{b} \lp r_x \rp }{t_{0x}}$ (and therefore $f_{b} \lp r_x \rp$  itself as well) have no local optima that at the same time are not global. The same trend continues for any $t_{0x}$ that we tested. This basically indicates continuous (over $t_{0x})$ presence of well-shaped loss landscape amenable to the use of descending algorithms. Whether or not other intrinsic features beyond the loss landscape play much of an additional role remains to be seen. Studies regarding organizational structuring of  ``near solutions'' might be of interest to pursue as next steps in these directions. A couple of very popular options include the overlap gap properties (OGP) \cite{Gamar21,GamarSud14,GamarSud17,GamarSud17a,AchlioptasCR11,HMMZ08,MMZ05} or local entropies (LE) \cite{Bald15,Bald16,Bald20}. Nonetheless, even if additional intrinsic properties do impact algorithmic performance,  the favorable -- unwanted local optima free -- loss landscapes that we uncover here are generically a \emph{necessary} condition for the success of descending algorithms. As such their absence must be ruled out to ensure generic solvability. Figure \ref{fig:fig2}-\ref{fig:fig4} demonstrate that this is the case for the first and second level of lifting. While tiny corrections on higher levels are present they are visually almost undetectable and make no significant impact on overall landscape smoothness phenomenology.
\begin{figure}[h]
\centering
\centerline{\includegraphics[width=1.00\linewidth]{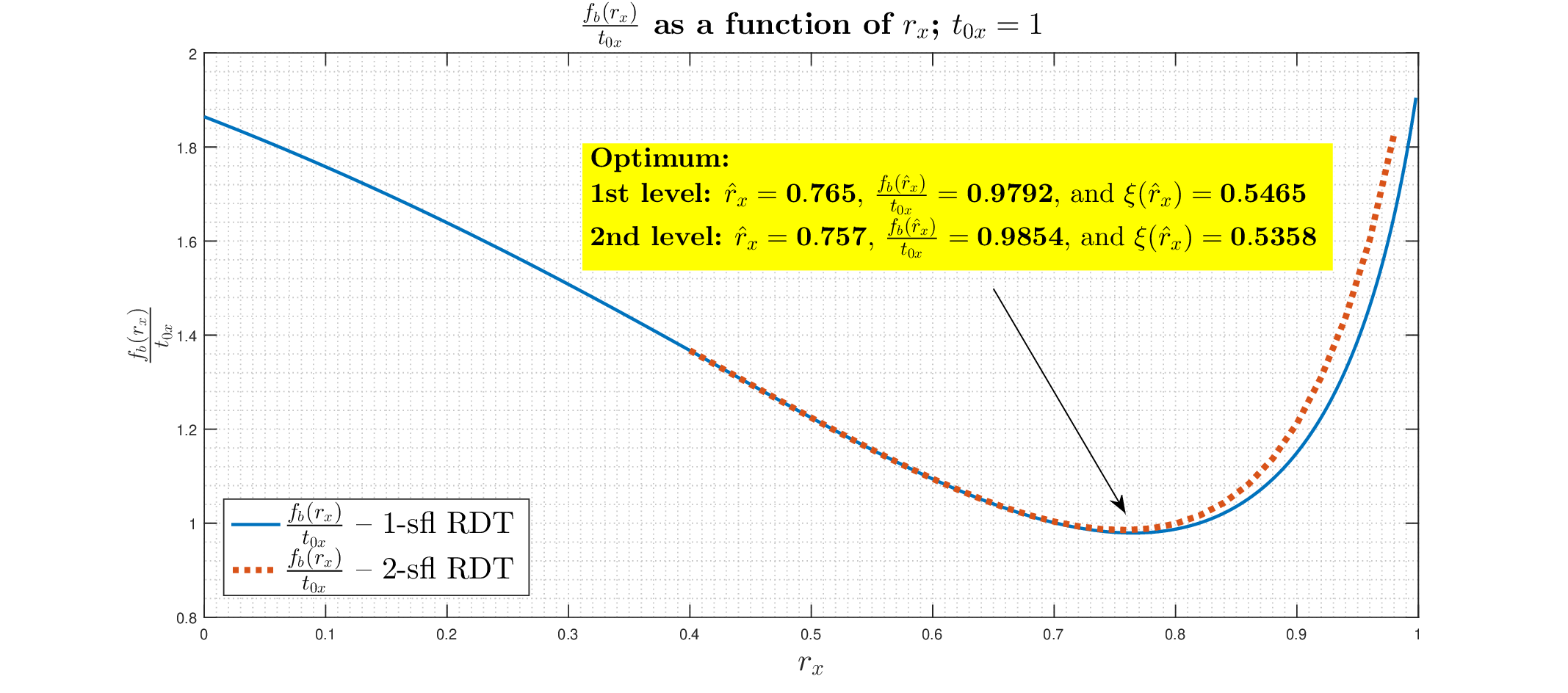}}
\caption{$\frac{f_b(r_x)}{t_{0x}}$ as a function of $r_x$; $t_{0x}=1$}
\label{fig:fig2}
\end{figure}
\begin{figure}[h]
\centering
\centerline{\includegraphics[width=1.00\linewidth]{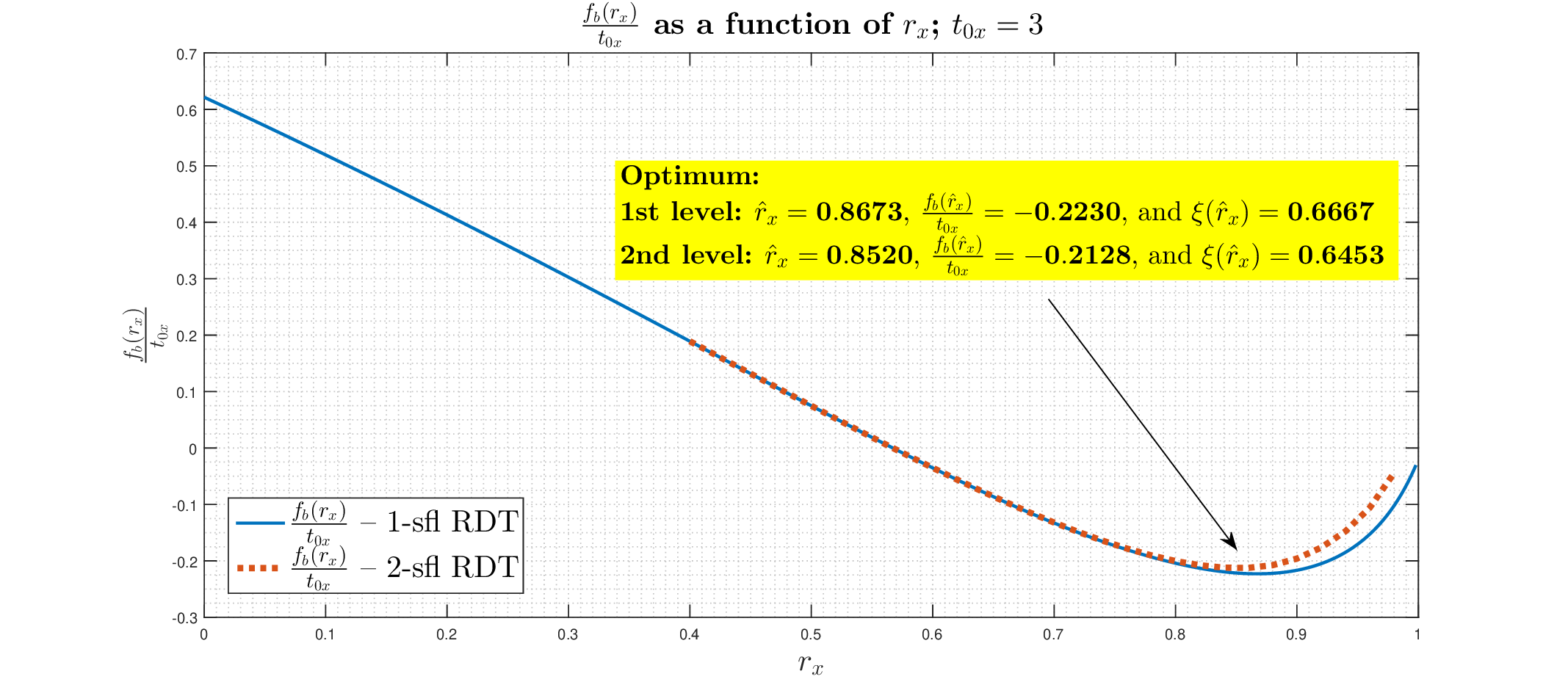}}
\caption{$\frac{f_b(r_x)}{t_{0x}}$ as a function of $r_x$; $t_{0x}=3$}
\label{fig:fig3}
\end{figure}
\begin{figure}[h]
\centering
\centerline{\includegraphics[width=1.00\linewidth]{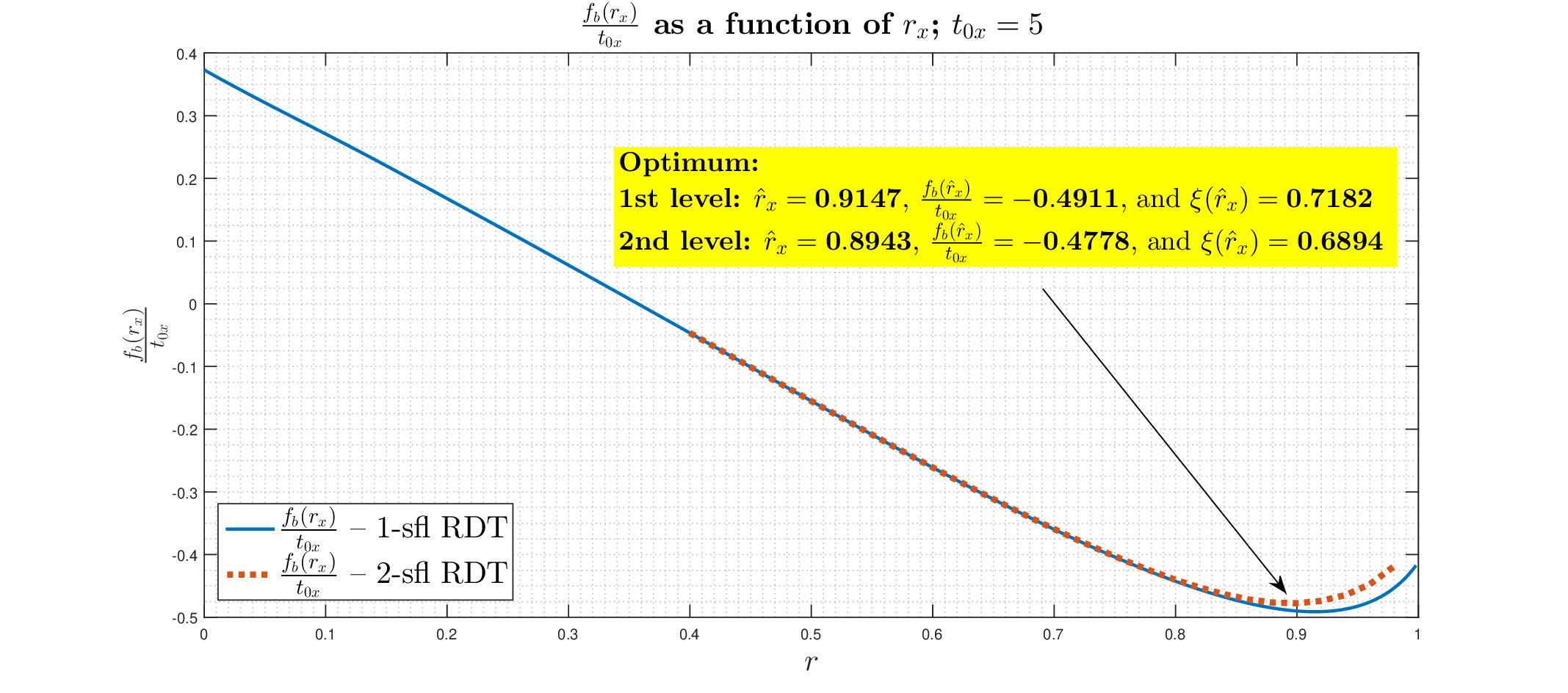}}
\caption{$\frac{f_b(r_x)}{t_{0x}}$ as a function of $r_x$; $t_{0x}=5$}
\label{fig:fig4}
\end{figure}

We also ran the following (barrier descending) procedure to simulate $f_{b} \lp \hat{r}_x \rp$
\begin{eqnarray}\label{eq:algimpeq10}
\x^{(t+1)} & \rightarrow  & \mbox{\textbf{gradbar}}\lp f_{b,1} \lp \x;\bar{t}_{0x}^{(t)}  \rp ;\x^{(t)},\bar{t}_{0x}^{(t)} \rp \nonumber \\
\bar{t}_{0x}^{(t+1)} & \rightarrow & \bar{c}^{(t)}\bar{t}_{0x}^{(t)},
\end{eqnarray}
with  function $f_{b,1} \lp \x;\bar{t}_{0x}^{(t)}  \rp $ specified by an argument $\bar{t}_{0x}^{(t)}$
\begin{eqnarray}\label{eq:algimpeq11}
f_{b,1} \lp\x;\bar{t}_{0x}^{(t)}\rp
 & =  &
  \bar{t}_{0x}^{(t)} \lp -t_{0x}  \|\x\|_2 - \log\lp - \lp \x^T  \lp 0.9 I  -  \frac{1}{2\sqrt{2n}} \lp G^T+G\rp   \rp \x - \kappa \rp \rp \rp
\nonumber \\
& & -\frac{1}{n}\sum_{i=1}^{n}  \log(1-n\x_i^2).
\end{eqnarray}

The results are shown in Figures \ref{fig:fig5}-\ref{fig:fig7}. The theoretical predictions are  given for the third partial level of lifting so that we do not expect basically any visible improvement on higher levels. Even though simulations are done with relatively small dimensions ($n=2000$), they are in an excellent agreement with the theoretical predictions for all three critical quantities, $\frac{f_b(\hat{r}_x)}{t_{0x}}$, $\xi(\hat{r}_x)$, and $\hat{r}_x$. In parallel we in Figures \ref{fig:fig5}-\ref{fig:fig7} show how these very same quantities behave when the original CLuP-SK dynamics is employed (to showcase that CLuP-SK dynamics excellent behavior starts for very small dimensions we chose three different values $n=200$, $n=1000$, and $n=2000$). Since this dynamics has a slightly different objective the curves are also different from the ones obtained for $\hat{r}_x$. However, as we hinted earlier, the phenomenology of all curves is completely preserved even though we used closely related objective for theoretical evaluations. Needless to say, for $t_{0x}\rightarrow \infty$ both dynamics converge to the same (optimal) values.

\begin{figure}[h]
\centering
\centerline{\includegraphics[width=1.00\linewidth]{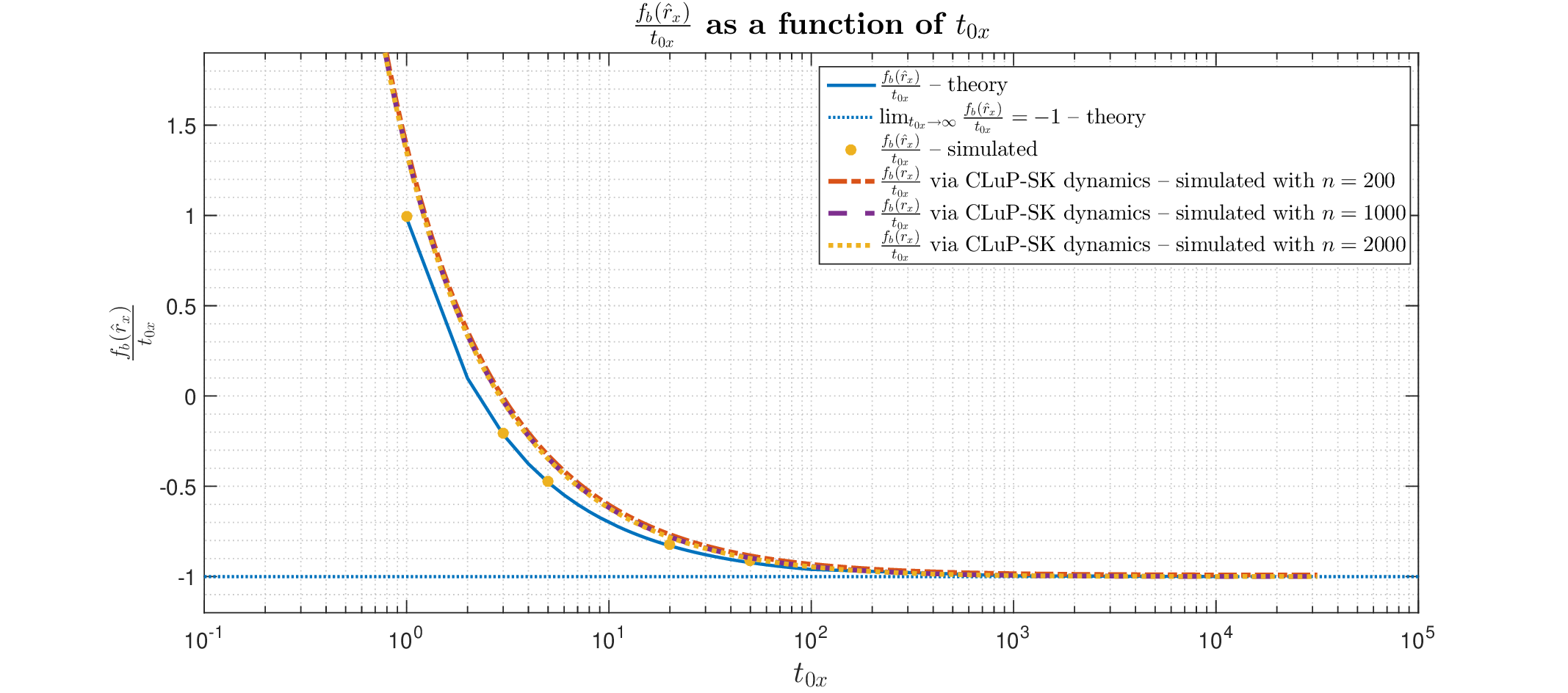}}
\caption{$\frac{f_b(\hat{r}_x)}{t_{0x}}$ as a function of $t_{0x}$}
\label{fig:fig5}
\end{figure}

\begin{figure}[h]
\centering
\centerline{\includegraphics[width=1.00\linewidth]{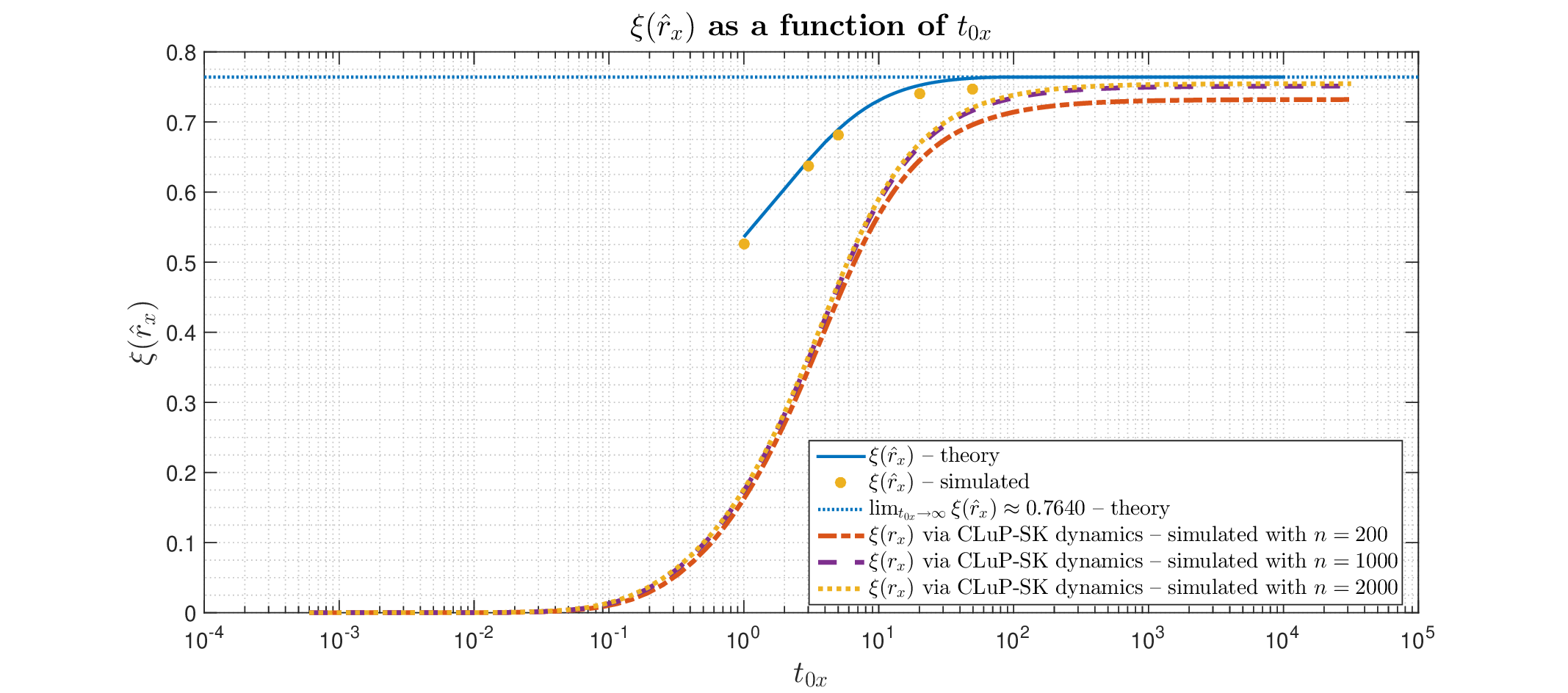}}
\caption{$\xi(\hat{r}_x)$ as a function of $t_{0x}$}
\label{fig:fig6}
\end{figure}

\begin{figure}[h]
\centering
\centerline{\includegraphics[width=1.00\linewidth]{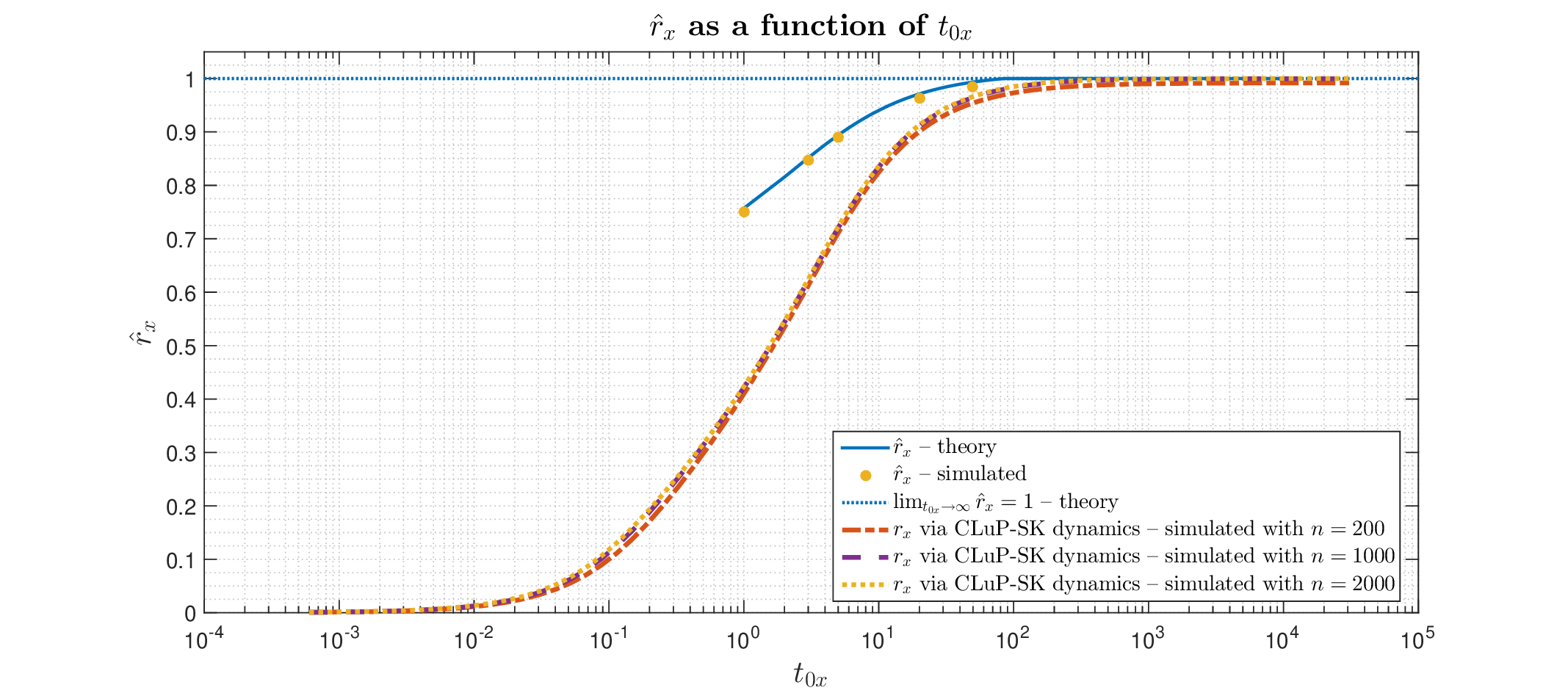}}
\caption{$\hat{r}_x$ as a function of $t_{0x}$}
\label{fig:fig7}
\end{figure}

\subsubsection{Numerically based \emph{exact} landscape characterization}
\label{sec:algimplosslandexact}

We above showcased a methodology that can be used to \emph{approximatively} characterize CLuP-SK loss landscape while avoiding potential numerical problems. Instead of focusing on true CLuP-SK objective, $\bar{f}_b \lp\x;\bar{t}_{0x}^{(t)}\rp$,, one focuses on its a slightly simpler ``trimmed'' version $\bar{f}_{b,1} \lp\x;\bar{t}_{0x}^{(t)}\rp$. Consequently, the resulting characterization is approximative but associated numerical evaluations are much simpler and more accurate.

As mentioned above, an alternative to the above approximative approach is to analyze $\bar{f}_b \lp\x;\bar{t}_{0x}^{(t)}\rp$ itself. Since the close-form helpful relations are not present in such an approach one has to completely rely on corresponding numerical evaluations. We show next what kind of results are obtained if such a route is pursued (needless to say, due to their numerical origins, these results need to be taken with a bit of additional caution). To that end, for  fixed $r_x$ ($0<r_x\leq 1$) and $\bar{r}_x$  we introduce the following
\begin{eqnarray}\label{eq:excclupskeq1}
\hspace{-1.5in} \bl{\overline{\mbox{\textbf{\emph{CLuP-SK}}}}} \mbox{\bl{\textbf{\emph{ model:}}}}  \hspace{.5in}  \max_{\x\in \bar{\cX}(r_x,\bar{r}_x)} \x^TG\x,
 \end{eqnarray}
with
\begin{eqnarray}\label{eq:excclupskeq1}
 \bar{\cX}(r_x,\bar{r}_x) \triangleq \left \{ \x | \x\in\mR^n,\|\x\|_2=r_x, \frac{1}{n}\sum_{i=1}^{n}\log\lp 1-n\x_i^2\rp =\bar{r}_x \right \}.
\end{eqnarray}
Analogously to (\ref{eq:cska0ham1})-(\ref{eq:cska0partfun}) we  define corresponding Hamiltonian and partition function
\begin{equation}
\bar{\cH}_{csk}(G)= \sum_{\x\in  \bar{\cX}(r_x,\bar{r}_x)   }\x^TG\x,\label{eq:exccska0ham1}
\end{equation}
and
\begin{equation}
\bar{Z}_{csk}(\beta,G)=\sum_{\x\in  \bar{\cX}(r_x,\bar{r}_x)   }e^{\beta \bar{\cH}_{csk}(G)}.\label{eq:exccska0partfun}
\end{equation}
Following closely  (\ref{eq:cska0logpartfunsqrt})-(\ref{eq:cska0limlogpartfunsqrt}) we then have for the thermodynamic limit average free energy
\begin{equation}
\bar{f}_{csk}(\beta)=\lim_{n\rightarrow\infty}\frac{\mE_G\log{(\bar{Z}_{csk}(\beta,G)})}{\beta \sqrt{2n}}
=\lim_{n\rightarrow\infty} \frac{\mE_G\log{(\sum_{\x\in \bar{\cX}(r_x,\bar{r}_x)} e^{\beta\bar{\cH}_{csk}(G)})}}{\beta \sqrt{2n}},\label{eq:exccska0logpartfunsqrt}
\end{equation}
and its ground state
\begin{eqnarray}
\bar{\xi}(r_x,\bar{r}_x) \triangleq \bar{f}_{csk}(\infty)  & \triangleq  & \lim_{\beta\rightarrow\infty} \bar{f}_{csk}(\beta)  =
\lim_{\beta,n\rightarrow\infty}\frac{\mE_G\log{(\bar{Z}_{csk}(\beta,G)})}{\beta \sqrt{2n}}
\nonumber \\
& = &
 \lim_{n\rightarrow\infty}\frac{\mE_G \max_{\x\in \bar{\cX} (r_x,\bar{r}_x)} \bar{\cH}_{csk}(G)}{\sqrt{2n}}  = \lim_{n\rightarrow\infty}\frac{\mE_G \max_{\x\in \bar{\cX} (r_x,\bar{r}_x)}\x^TG\x}{\sqrt{2n}}.
  \label{eq:exccska0limlogpartfunsqrt}
\end{eqnarray}
Utilizing $\bar{\cX}(r_x,\bar{r}_x)$ instead of $\cX(r_x)$ and mimicking step-by-step the procedure presented in Section \ref{sec:randlincons} we now have analogously to (\ref{eq:excprac11})
 \begin{eqnarray}
\bar{\xi}(r_x,\bar{r}_x) = \bar{f}_{csk}(\infty)
& = &  \lim_{n\rightarrow\infty}\frac{\mE_G \max_{\x\in \bar{\cX}(r_x\bar{r}_x)} \x^TG\x  }{\sqrt{2n}}
    =
 -\frac{1}{\sqrt{2}}\lim_{n\rightarrow\infty} \psi_{rd,1}(\hat{\q},\hat{\c},r_x,\bar{r}_x),
  \label{eq:excprac11}
\end{eqnarray}
where following  (\ref{eq:prac1}) and (\ref{eq:prac3})
\begin{align}\label{eq:excprac1}
    \psi_{rd,1}(\q,\c,r_x,\bar{r}_x)
  & =   \frac{1}{2}    \sum_{k=2}^{r+1}\Bigg(\Bigg.
   \q^2_{k-1}
   -\q^2_{k}
  \Bigg.\Bigg)
\c_k
  - \frac{1}{n}\varphi(\bar{D}^{(bin)}(r_x,\bar{r}_x)),
  \end{align}
non-fixed parts of $\hat{\q}$, and  $\hat{\c}$  are the solutions of
\begin{eqnarray}\label{eq:excprac2}
    \frac{d \psi_{rd,1}(\q,\c,r_x,\bar{r}_x)}{d\q} & = & 0 \nonumber \\
   \frac{d \psi_{rd,1}(\q,\c,r_x,\bar{r}_x)}{d\c} & = & 0,
   \end{eqnarray}
 and
\begin{eqnarray}\label{eq:excprac3}
\bar{D}^{(bin)}(r_x,\bar{r}_x) & = & \max_{\x\in\bar{\cX}(r_x.\bar{r}_x)} \lp  r_x \sqrt{2n}      \lp\sum_{k=2}^{r+1}c_k\h^{(k)}\rp^T\x  \rp.
 \end{eqnarray}
We then find
\begin{eqnarray}\label{eq:excprac4}
\bar{D}^{(bin)}(r_x,\bar{r}_x)  =  \max_{\x\in\bar{\cX}(r_x,\bar{r}_x)} \lp r_x  \sqrt{2n}      \lp\sum_{k=2}^{r+1}c_k\h^{(k)}\rp^T\x  \rp
 = \min_{\gamma,\nu} \lp - r_x \sum_{i=1}^n \bar{D}^{(bin)}_i(c_k) +\gamma r_x^3 n +\nu r_x \bar{r}_x n \rp,
 \end{eqnarray}
where we have adopted scaling $\gamma\sim \gamma\sqrt{n}$ and $\nu\sim \nu\sqrt{n}$ and
\begin{eqnarray}\label{eq:excprac5}
\bar{D}^{(bin)}_i(c_k)=    - \sqrt{2}      \lp\sum_{k=2}^{r+1}c_k\h_i^{(k)}\rp \x_i +\gamma\x_i^2 + \nu\log(1-n\x_i^2).
\end{eqnarray}
From (\ref{eq:excprac1})-(\ref{eq:excprac5})
 \begin{eqnarray}
 \lim_{n\rightarrow\infty} \psi_{rd,1}(\hat{\q},\hat{\c},r_x,\bar{r}_x) =  \bar{\psi}_{rd,1}(\hat{\q},\hat{\c},\hat{\gamma},\hat{\nu},r_x,\bar{r}_x),
  \label{eq:exvprac12a}
\end{eqnarray}
where
\begin{eqnarray}\label{eq:excprac13}
    \bar{\psi}_{rd,1}(\q,\c,\gamma_{sq},r_x)   & = &  \frac{1}{2}    \sum_{k=2}^{r+1}\Bigg(\Bigg.
   \q^2_{k-1}
   - \q^2_{k}
  \Bigg.\Bigg)
\c_k
-\gamma r_x^3 -\nu r_x\bar{r}_x
- \varphi(\bar{D}_1^{(bin)}(c_k(\q)),\c),
  \end{eqnarray}
  and analogously to (\ref{eq:prac14})
 \begin{align}\label{eq:excprac14}
\varphi(\bar{D}_1^{(bin)}(c_k(\q)),\c) & =
 \mE_{{\mathcal U}_{r+1}} \frac{1}{\c_r} \log
\lp \mE_{{\mathcal U}_{r}} \lp \dots \lp \mE_{{\mathcal U}_3}\lp\lp\mE_{{\mathcal U}_2} \lp
    e^{  -r_x \c_2 \bar{D}_1^{(bin)}(c_k(\q))  }  \rp\rp^{\frac{\c_3}{\c_2}}\rp\rp^{\frac{\c_4}{\c_3}} \dots \rp^{\frac{\c_{r}}{\c_{r-1}}}\rp, \nonumber \\
  \end{align}
with $\hat{\gamma}$, $\hat{\nu}$, and  the non-fixed parts of $\hat{\q}$, and  $\hat{\c}$ being the solutions of
\begin{eqnarray}\label{eq:excprac16}
    \frac{d \bar{\psi}_{rd,1}(\q,\c,\gamma,\nu,r_x,\bar{r}_x)}{d\q} & = & 0 \nonumber \\
   \frac{d \bar{\psi}_{rd,1}(\q,\c,\gamma,\nu,r_x,\bar{r}_x)}{d\c} & = & 0 \nonumber \\
   \frac{d \bar{\psi}_{rd,1}(\q,\c,\gamma,\nu,r_x,\bar{r}_x)}{d\gamma} & = & 0  \nonumber \\
   \frac{d \bar{\psi}_{rd,1}(\q,\c,\gamma,\nu,r_x,\bar{r}_x)}{d\nu} & = & 0. \end{eqnarray}
One then also observes
\begin{eqnarray}\label{eq:excprac17}
c_k(\hat{\q})  & = & \sqrt{\hat{\q}_{k-1}-\hat{\q}_k},
 \end{eqnarray}
and after connecting  (\ref{eq:excprac11}) to (\ref{eq:excprac13}) and (\ref{eq:excprac14}) obtains
 \begin{eqnarray}
\bar{\xi}(r_x,\bar{r}_x)  = \bar{f}_{csk}(\infty)
& = &  \lim_{n\rightarrow\infty}\frac{\mE_G \max_{\x\in\bar{\cX}(r_x,\bar{r}_x)} \x^TG\x}{\sqrt{2n}}  =
 -\frac{1}{\sqrt{2}}\lim_{n\rightarrow\infty} \psi_{rd,1}(\hat{\q},\hat{\c},r_x,\bar{r}_x)
 \nonumber \\
&  = &  - \frac{1}{\sqrt{2}} \bar{\psi}_{rd,1}(\hat{\q},\hat{\c},\hat{\gamma},\hat{\nu},r_x,\bar{r}_x) \nonumber \\
 & = & \frac{1}{\sqrt{2}}\lp  -\frac{1}{2}    \sum_{k=2}^{r+1}\Bigg(\Bigg.
    \hat{\q}^2_{k-1}
   - \hat{\q}^2_{k}
  \Bigg.\Bigg)
\hat{\c}_k
 + \hat{\gamma} r_x^3 + \hat{\nu}r_x\bar{r}_x
  + \varphi(\bar{D}_1^{(bin)}(c_k(\hat{\q})),\c) \rp. \nonumber \\
  \label{eq:excprac18}
\end{eqnarray}
We summarize the above results in the following theorem.

\begin{theorem}
  \label{thme:thmprac2}
Assume the setup of of Theorem \ref{thme:thmprac1}  with $\varphi(\cdot)$ and $\bar{\psi}_{rd,1}(\cdot)$  as in (\ref{eq:excprac14}) and (\ref{eq:excprac13}), respectively. Let the ``fixed'' parts of $\hat{\q}$, and $\hat{\c}$ be $\hat{\q}_1\rightarrow 1$, $\hat{\c}_1\rightarrow 1$, $\hat{\q}_{r+1}=\hat{\c}_{r+1}=0$. Also, let $\hat{\gamma}$, $\hat{\nu}$, and the ``non-fixed'' parts of $\hat{\q}_k$, and $\hat{\c}_k$ ($k\in\{2,3,\dots,r\}$) be the solutions of (\ref{eq:excprac16}). For $c_k(\hat{\q})$  as in (\ref{eq:excprac17}), one has
 \begin{eqnarray}
\bar{\xi}(r_x,\bar{r}_x) \triangleq \bar{f}_{csk}(\infty)
& = &   \frac{1}{\sqrt{2}} \lp -\frac{1}{2}    \sum_{k=2}^{r+1}\Bigg(\Bigg.
    \hat{\q}^2_{k-1}
   - \hat{\q}^2_{k}
  \Bigg.\Bigg)
\hat{\c}_k
+ \hat{\gamma} r_x^3 + \hat{\nu} r_x \bar{r}_x
  + \varphi(\bar{D}_1^{(bin)}(c_k(\hat{\q})),\hat{\c})   \rp. \nonumber \\
  \label{eq:excthmprac1eq1}
\end{eqnarray}
\end{theorem}
\begin{proof}
Follows from the above discussion, Theorems \ref{thm:thmsflrdt1} and \ref{thme:thmprac1}, and the sfl RDT machinery presented in \cite{Stojnicnflgscompyx23,Stojnicsflgscompyx23,Stojnicflrdt23}.
\end{proof}

The above theorem is in principle sufficient to analyze $\bar{f}_{b,x} \lp\x;\bar{t}_{0x}^{(t)}\rp$. We first recall
\begin{eqnarray}\label{eq:excalgimpeq6}
\bar{f}_{b,x} \lp\x; t_{0x} \rp = - t_{0x} \|\x\|_2 - \log\lp - \lp \x^T  \lp 0.9 I  -  \frac{1}{2\sqrt{2n}} \lp G^T+G\rp   \rp \x - \kappa \rp \rp
- \frac{1}{n}\sum_{i=1}^{n} \log\lp 1-n\x_i^2  \rp.
\end{eqnarray}
For a fixed $t_{0x}$ we are interested in behavior of
\begin{equation}\label{eq:excalgimpeq7}
\bar{f}_{b} \lp r_x,\bar{r}_x \rp = \min_{\x\in\bar{\cX}(r_x,\bar{r}_x)}
  - t_{0x} \|\x\|_2 - \log\lp - \lp \x^T  \lp 0.9 I  -  \frac{1}{2\sqrt{2n}} \lp G^T+G\rp   \rp \x - \kappa \rp \rp
- \frac{1}{n}\sum_{i=1}^{n} \log\lp 1-n\x_i^2  \rp .
\end{equation}
Recalling on (\ref{eq:excclupskeq1}) and (\ref{eq:exccska0limlogpartfunsqrt})
\begin{eqnarray}\label{eq:excalgimpeq8}
\bar{f}_{b} \lp r_x,\bar{r}_x \rp
& = & \min_{\x\in,\bar{\cX}(r_x,\bar{r}_x )} - t_{0x}r_x - \log \lp  -0.9r_x^2 + \bar{\xi}(r_x,\bar{r}_x ) + \kappa \rp -\bar{r}_x
\nonumber \\
& = &
 - t_{0x}r_x - \log \lp  -0.9r_x^2 + \bar{f}_{csk}(\infty)  + \kappa \rp - \bar{r}_x .
\end{eqnarray}
We set
\begin{eqnarray}\label{eq:excalgimpeq9a0a0}
\bar{r}_x^{(opt)}  \triangleq  \mbox{argmin}_{\bar{r}_x <0} \bar{f}_{b} \lp r_x,\bar{r}_x \rp,
\end{eqnarray}
and
\begin{eqnarray}\label{eq:excalgimpeq9a0}
\bar{f}_{b} \lp r_x \rp
 &  \triangleq &  \min_{\bar{r}_x <0} \bar{f}_{b} \lp r_x,\bar{r}_x \rp  =  \bar{f}_{b} \lp r_x,\bar{r}_x^{(opt)} \rp \nonumber \\
\bar{\xi} \lp r_x \rp
 &  \triangleq &    \bar{\xi} \lp r_x,\bar{r}_x^{(opt)} \rp.
\end{eqnarray}
Analogously to (\ref{eq:algimpeq9}) we now have for the optimal $r_x$
\begin{eqnarray}\label{eq:excalgimpeq9}
\hat{r}_x =  \mbox{argmin}_{r_x\in(0,1]} \bar{f}_{b} \lp r_x \rp.
\end{eqnarray}
In Figure \ref{fig:fig2a0}, we show $\frac{\bar{f}_{b} \lp r_x \rp }{t_{0x}}$ for $t_{0x}=20$. Results obtained on the third partial level of  lifting are used  for $\bar{\xi}(r_x,\bar{r}_x) = \bar{f}^{(3,p)}_{csk} (\infty) $. One observes that  $\frac{\bar{f}_{b} \lp r_x \rp }{t_{0x}}$ (and therefore $\bar{f}_{b} \lp r_x \rp$  itself as well) have no local optima different from global ones. The same phenomenological behavior  continues for any $t_{0x}$ that we tested, basically indicating presence of descending algorithms favorable landscape. Discussion from previous sections regarding the role of other intrinsic features (like OGP or local entropies) applies here as well.
\begin{figure}[h]
\centering
\centerline{\includegraphics[width=1.00\linewidth]{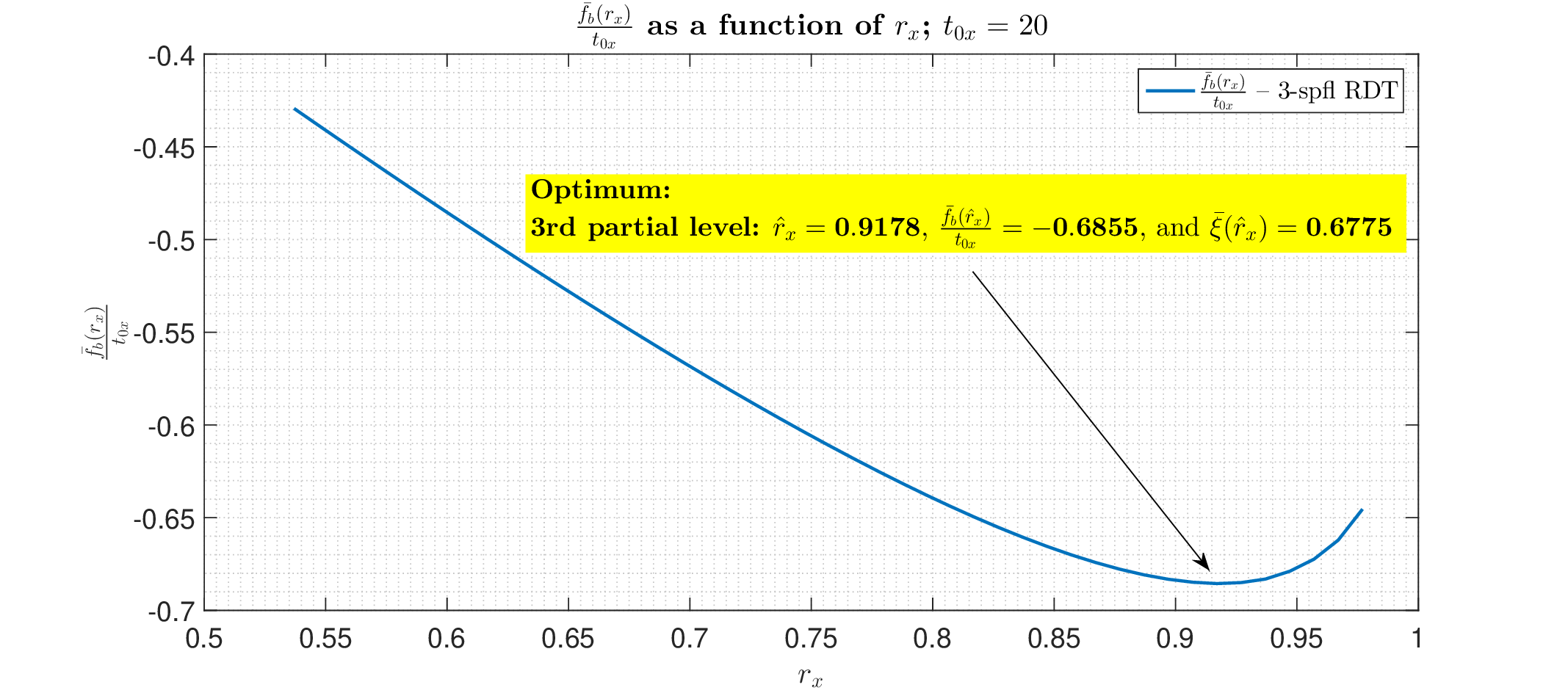}}
\caption{$\frac{\bar{f}_b(r_x)}{t_{0x}}$ as a function of $r_x$; $t_{0x}=20$ -- $\overline{\mbox{CLuP-SK}} \mbox{ model}$ }
\label{fig:fig2a0}
\end{figure}

In Figures \ref{fig:fig5a0}-\ref{fig:fig7a0} theoretical predictions for all three critical quantities, $\frac{\bar{f}_b(\hat{r}_x)}{t_{0x}}$, $\bar{\xi}(\hat{r}_x)$, and $\hat{r}_x$ are shown. As earlier, we use the third partial level of lifting so that no further visible improvement on higher lifting levels is expected. The theoretical predictions are accompanied with the simulated results obtained by running CLuP-SK for $n=200$, $n=1000$, and $n=2000$. As can be seen from figures, even though the dimensions are fairly small (compared to $n\rightarrow \infty$), the agreement between theoretical and simulated results is excellent.

\begin{figure}[h]
\centering
\centerline{\includegraphics[width=1.00\linewidth]{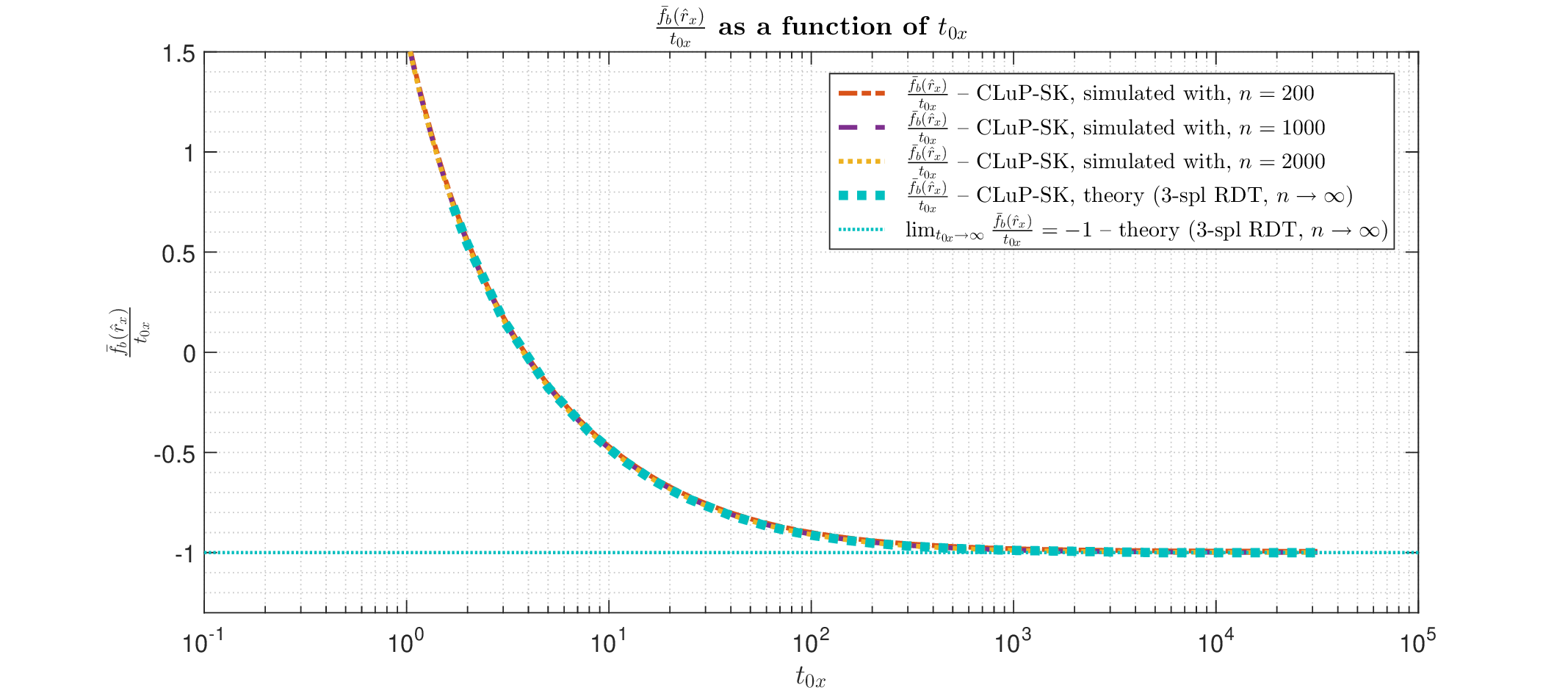}}
\caption{$\frac{\bar{f}_b(\hat{r}_x)}{t_{0x}}$ as a function of $t_{0x}$  -- $\overline{\mbox{CLuP-SK}} \mbox{ model}$ }
\label{fig:fig5a0}
\end{figure}
\begin{figure}[h]
\centering
\centerline{\includegraphics[width=1.00\linewidth]{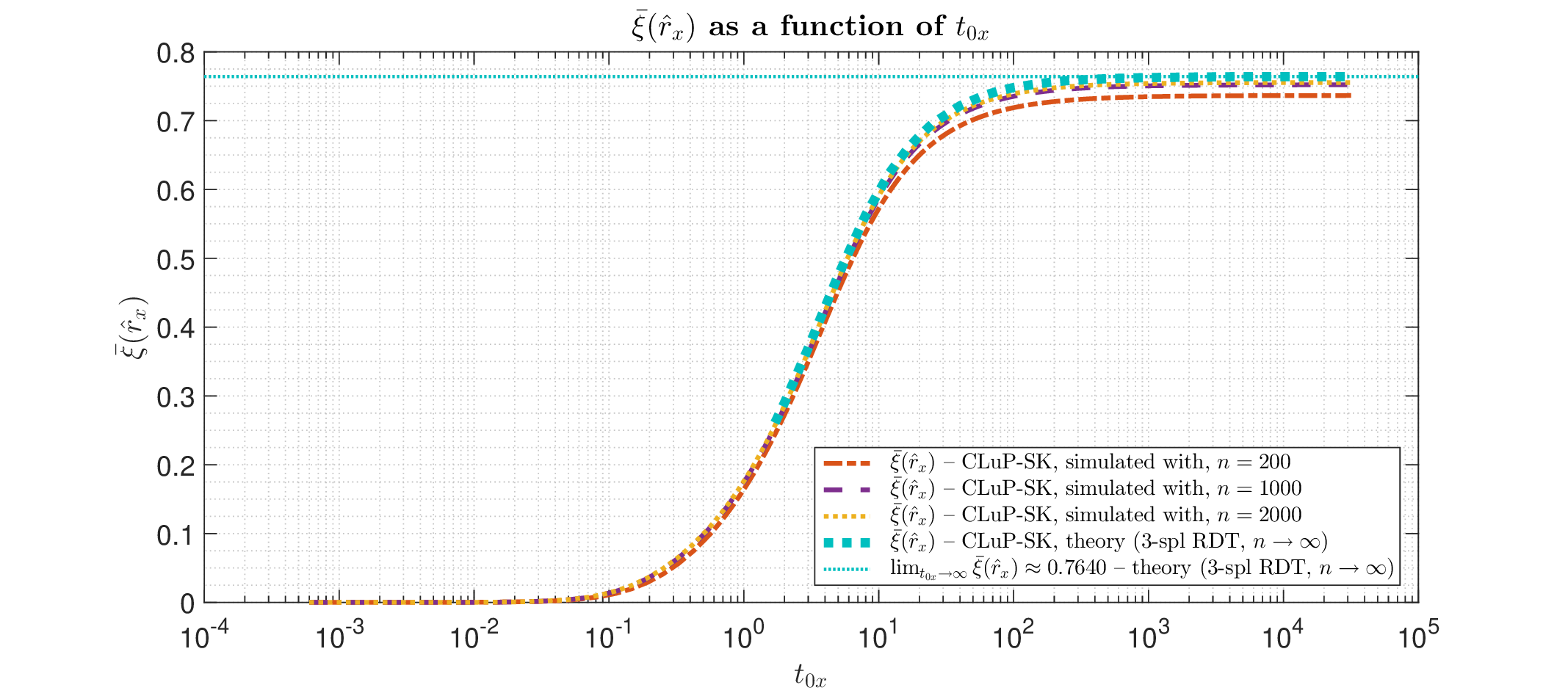}}
\caption{$\bar{\xi}(\hat{r}_x)$ as a function of $t_{0x}$ -- $\overline{\mbox{CLuP-SK}} \mbox{ model}$  }
\label{fig:fig6a0}
\end{figure}
\begin{figure}[h]
\centering
\centerline{\includegraphics[width=1.00\linewidth]{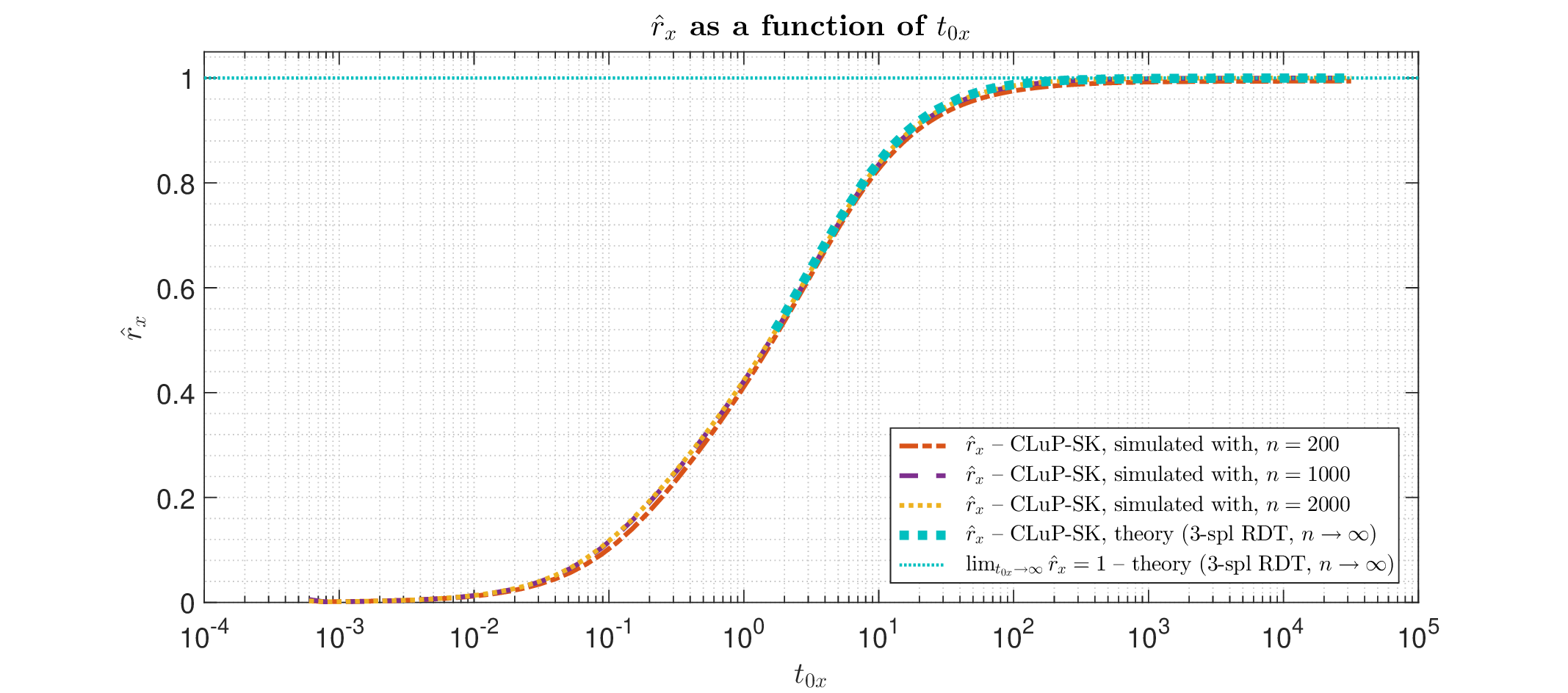}}
\caption{$\hat{r}_x$ as a function of $t_{0x}$  -- $\overline{\mbox{CLuP-SK}} \mbox{ model}$ }
\label{fig:fig7a0}
\end{figure}

In Figures \ref{fig:figconv1} and \ref{fig:figconv2} we show the effect that changing the underlying dimension $n$ has on $\bar{\xi}(\hat{r}_x)$.  Increasing $n$ from a few tens and hundreds to a few thousands  we observe at what pace the simulated CLuP-SK dynamics approaches theoretical predictions. In particular, we see that the most rapid part of the convergence process happens already for $n$ on the order of thousand.
\begin{figure}[h]
\centering
\centerline{\includegraphics[width=1.00\linewidth]{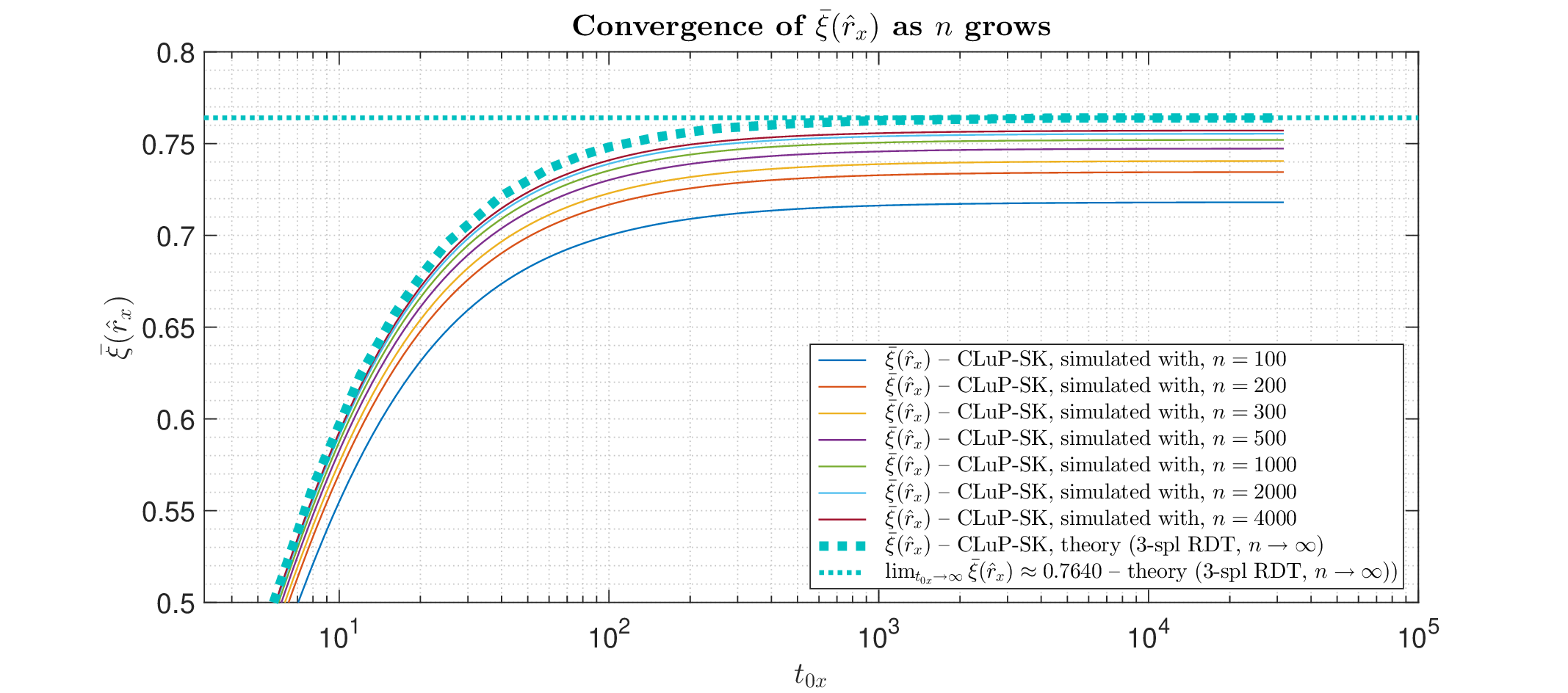}}
\caption{Convergence of $\bar{\xi}(\hat{r}_x)$ as $n$ grows -- $\overline{\mbox{CLuP-SK}} \mbox{ model}$  }
\label{fig:figconv1}
\end{figure}
\begin{figure}[h]
\centering
\centerline{\includegraphics[width=1.00\linewidth]{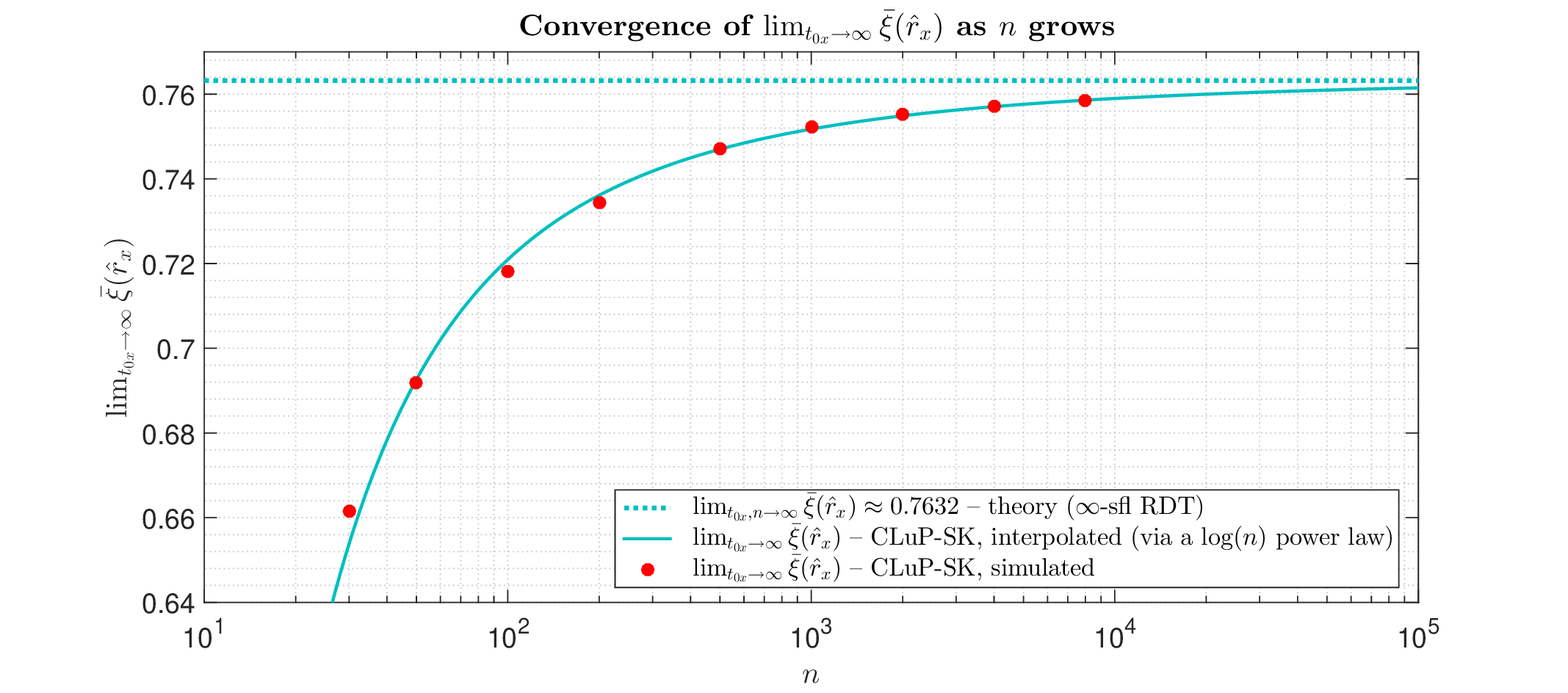}}
\caption{Convergence of $\lim_{t_{0x}\rightarrow\infty}\bar{\xi}(\hat{r}_x)$ as $n$ grows -D- $\overline{\mbox{CLuP-SK}} \mbox{ model}$  }
\label{fig:figconv2}
\end{figure}

\subsection{Overlaps and ultrametricity}
\label{sec:ovlpultmet}

Excellent performance of CLuP-SK procedure allows to simulate behavior of near optimal solutions -- spin configurations that produce close to optimal ground state free energies. One is particularly interested in the structure of overlaps, associated GIbbs measures, and ultrametricity. In the thermodynamic limit the key parts of the Gibbs measure associated with the overlaps concentrate on $\q_2,\q_3,\dots,\q_r$. In Figure \ref{fig:fig8} we show how $\q$ changes as the lifting process progresses. In particular, we associate with $\q$ and $\c$ the following
\begin{eqnarray}\label{eq:ultmet1}
\bl{\mbox{\textbf{$\q\lp \frac{\c}{\c_2}\rp $ map:}}}  \hspace{.7in} \q_2 \leftrightarrow   \left [ \frac{\c_3}{\c_2}, \frac{\c_2}{\c_2} \right ],\quad
\q_3  \leftrightarrow  \left [ \frac{\c_4}{\c_2}, \frac{\c_3}{\c_2}  \right ], \quad \dots.
\end{eqnarray}
Concrete numerical values for $\q$ and $\c$ up to 5-th partial level of lifting are given in Table \ref{tab:5rsbSK}. In parallel with $\q\lp \frac{\c}{\c_2}\rp $ map, Figure \ref{fig:fig8} also shows near optimum overlap values obtained utilizing CLuP-SK algorithmic procedure with $n=2000$. It is interesting to note that even though we focused only on configurations that are very close to the optimum, already on the fifth level of lifting the $\q\lp \frac{\c}{\c_2}\rp $ Gibbs measure cdf to a large degree matches the overlap distribution obtained through simulations.

For the completeness we also include the cdf predictions obtained via replica methods. As is well known, the original replica symmetry breaking (RSB) structure was invented by Parisi in \cite{Par79,Par80,Par83,Parisi80}. Early numerical considerations included first two steps of breaking as demonstrations of the overall RSB power. Moreover, for a majority of key quantities that describe the SK model behavior (including the most popular ground state free energy), accuracy achieved after the first two RSB steps for all practical purposes sufficed. However, the cdf of the Gibbs measure is a very important exception where obtaining even remotely accurate results requires a fairly high number of RSB steps. Consequently, evaluations via small number of RSB steps were soon abandoned and replaced by  different alternatives. Continuous domain formulations via differential equations became particularly popular. An excellent set of results in this direction was obtained in \cite{CrisRizo02} where  $0.76321\pm 0.00003$ was obtained as the SK-model ground state free energy (this closely matched Parisi's original $0.7633\pm0.0001$). Interestingly and somewhat paradoxically,  \cite{OppSS07,OppSch08,OppShe05,SchOpp08} made a switch back and used original discrete (finite number of steps) RSB formulation to obtain approximate predictions up to $200$-RSB steps. In particular, \cite{OppSch08,OppSS07} gave $\approx 0.76317$ as the SK's ground state free energy. We found as particularly simple and elegant the $\infty$-RSB  approximation $\q_{app}(\c)=\frac{\sqrt{pi}}{2}\frac{\c}{\xi_{app}}\erf\lp \frac{\xi_{app}}{\c}\rp $ with $\xi_{app} \approx 1.13\approx \frac{2}{\sqrt{\pi}}$ given in \cite{OppShe05}. As demonstrated on a multitude of occasions in \cite{OppSS07,OppSch08,SchOpp08}, despite its simplicity, $\q_{app}(\c)$ very closely approximates the above mentioned $200$-RSB predictions. Its a normalized variant (to account for $\frac{\c}{\c_2}$ scaling) is shown in Figure \ref{fig:fig8} as well.

In Figure \ref{fig:fig9} we show the Gram matrix of overlaps obtained via the same CLuP-SK algorithm that we used above. Even for fairly small $n=2000$ and with a focus solely on the configurations that produce values very close to the optimum, one observes emergence of a beautiful ultrametric structure, precisely as the theory predicts.

\begin{figure}[h]
\centering
\centerline{\includegraphics[width=1.00\linewidth]{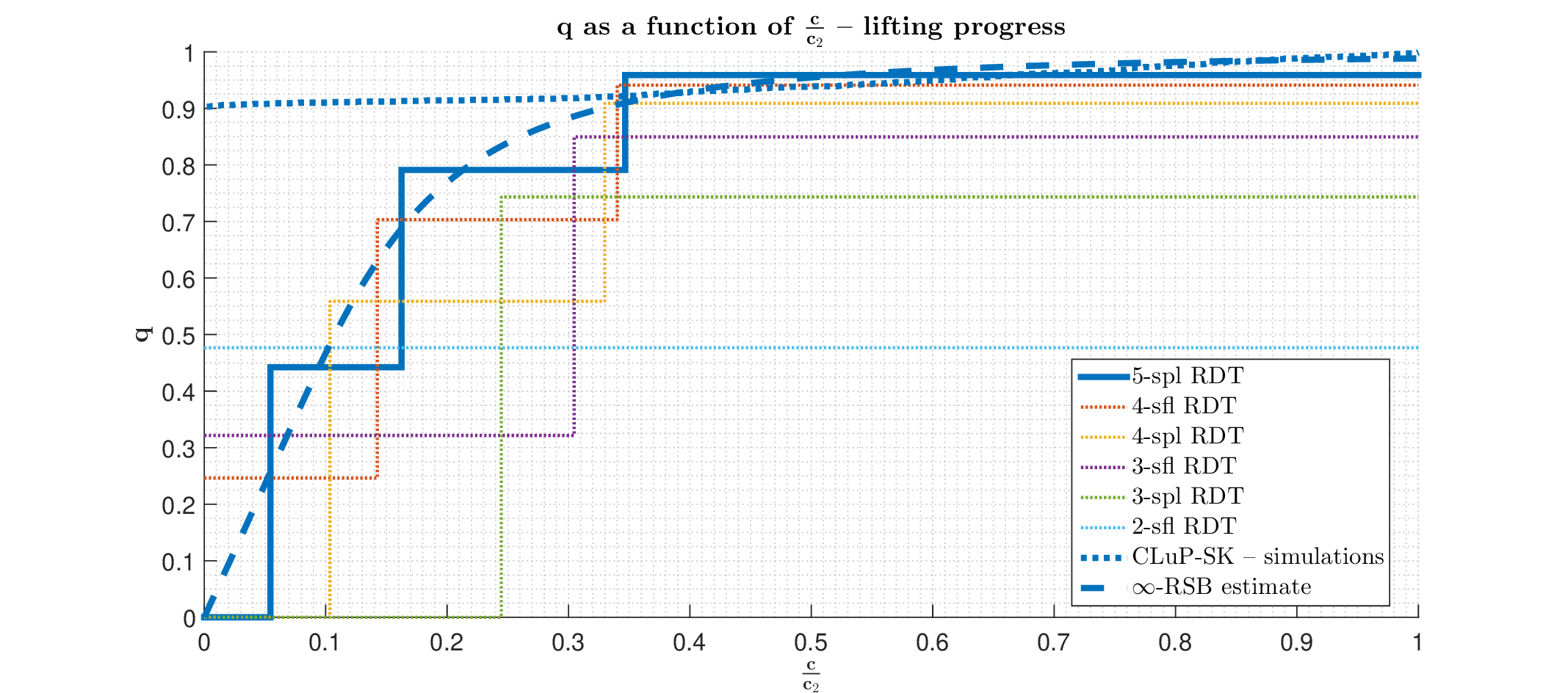} }
\caption{$\q$ as a function of $\frac{\c}{\c_2}$}
\label{fig:fig8}
\end{figure}

\begin{figure}[h]
\centering
\centerline{\includegraphics[width=.80\linewidth]{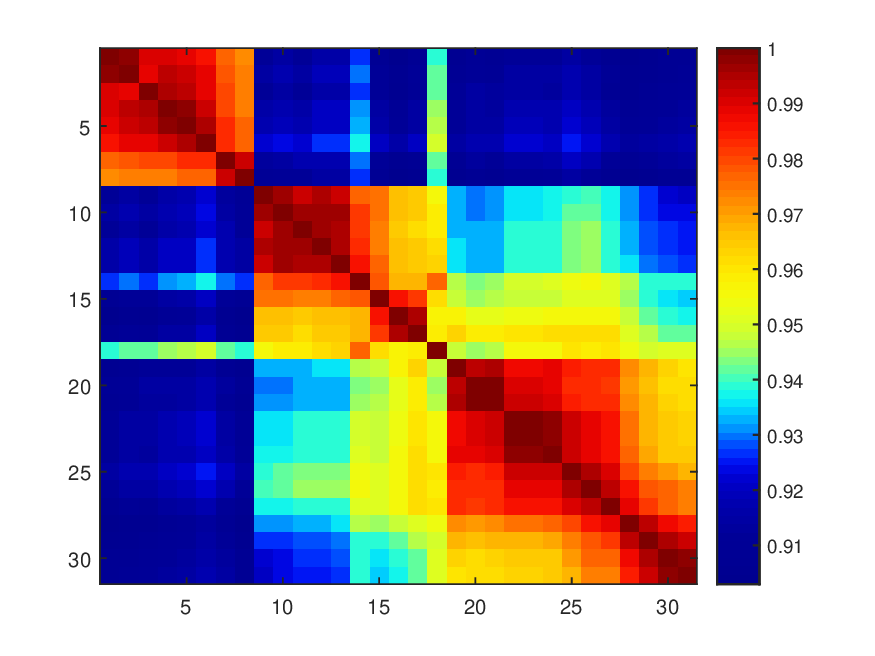}}
\caption{SK-model - Gram matrix of overlaps near the optimum -- emergence of ultrametricity}
\label{fig:fig9}
\end{figure}

\begin{table}[h]
\caption{$r$-sfl RDT parameters; CLUP SK model; $r_x=1$;  $\hat{\c}_1\rightarrow 1$; $n,\beta\rightarrow\infty$}\vspace{.1in}
\centering
\def\arraystretch{1.2}
\begin{tabular}{||l||c||c|c|c|c||c|c|c|c||c||}\hline\hline
 \hspace{-0in}$r$-sfl RDT                                             & $\hat{\gamma}$  &   $\hat{\q}_4$  & $\hat{\q}_3$  & $\hat{\q}_2$  & $\hat{\q}_1$ &  $\hat{\c}_5$   &  $\hat{\c}_4$   &   $\hat{\c}_3$   &   $\hat{\c}_2$    & $f_{csk}^{(r)}(\infty)$  \\ \hline\hline
$\mathbf{1}$ (full)                                      & $0$ & $0$ & $0$ & $0$ & $\rightarrow 1$ &  $\rightarrow 0$
 &  $\rightarrow 0$ &  $\rightarrow 0$  &  $\rightarrow 0$  & \bl{$\mathbf{0.79788}$} \\ \hline\hline
$\mathbf{2}$ (partial)                                      & $0$ & $0$ & $0$ & $0$ & $\rightarrow 1$ &  $\rightarrow 0$
 &  $\rightarrow 0$ &  $\rightarrow 0$  &  $   0.5779 $  & \bl{$\mathbf{0.76883}$} \\ \hline
   $\mathbf{2}$ (full)                                      & $0$  & $0$ & $0$ & $0.4768$ & $\rightarrow 1$ &  $\rightarrow 0$
 &  $\rightarrow 0$ &  $\rightarrow 0$
 &  $0.9623$   & \bl{$\mathbf{0.76526}$}  \\ \hline\hline
 $\mathbf{3}$ (partial)                                      & $0$ & $0$  & $0$ & $0.7434$ & $\rightarrow 1$ &  $\rightarrow 0$
 &  $\rightarrow 0$ &  $0.3569$
 &  $1.4586$   & \bl{$\mathbf{0.76403}$}   \\ \hline
 $\mathbf{3}$ (full)                                      & $0$ & $0$  & $0.3215$ & $0.8496$ & $\rightarrow 1$ &  $\rightarrow 0$ &  $\rightarrow 0$
&  $0.5906$
 &  $ 1.9386$   & \bl{$\mathbf{0.76361}$}   \\ \hline\hline
 $\mathbf{4}$ (partial)                                      & $0$ & $0$  & $ 0.5587$ & $ 0.9088$ & $\rightarrow 1$ &  $\rightarrow 0$
 &  $0.2599$ &  $  0.8280$
 &  $ 2.5103$   & \bl{$\mathbf{0.76341}$}   \\ \hline
 $\mathbf{4}$ (full)                                      & $0$ & $ 0.2462$  & $ 0.7031 $ & $0.9410$ & $\rightarrow 1$ &  $\rightarrow 0$
 &  $0.4469$ &  $1.0663$
 &  $   3.1345$   & \bl{$\mathbf{0.76331}$}   \\ \hline\hline
 $\mathbf{5}$ (partial)                                      & $0$ & $  0.4421$  & $ 0.7912 $ & $ 0.9588$ & $\rightarrow 1$  & $0.2048 $  &  $0.6104$ &  $1.3026$
 &  $ 3.7571$   & \bl{$\mathbf{0.76326}$}   \\ \hline\hline
\end{tabular}
\label{tab:5rsbSK}
\end{table}

\section{Conclusion}
\label{sec:conc}

We studied the algorithmic aspects of  Sherrington-Kirkpatrick (SK) spin glass model. Within the \emph{worst case} centered classical NP complexity theory the SK model's ground state free energy is hard to approximate within a $\log(n)^{const.}$ factor. On the other hand, exploiting the SK's random nature polynomial spectral methods \emph{typically} approach the optimum within a constant factor. Design of efficient optimization procedures with approximability factor arbitrarily close to 1 is a key algorithmic imperative.

To address such a challenge we devised a \emph{Controlled Loosening-up}  CLuP-SK algorithmic procedure. Associating to it a  (random) CLuP-SK and $\overline{\mbox{CLuP-SK}}$  models and utilizing fully lifted random duality theory (fl RDT) \cite{Stojnicflrdt23}, we developed a generic framework to characterize the algorithm's performance. Extensive numerical experiments demonstrated  an excellent agrement between the algorithm's practical behavior and the corresponding  theoretical predictions. Most notably, already for $n$ on the order of few thousands  CLuP-SK achieves $\sim 0.76$ ground state free energy (which remarkably closely approaches theoretical $n\rightarrow\infty$ limit $\approx 0.763$).

Generic nature of the introduced concepts  ensures that various generalizations and extensions can be done as well. Development of analogous algorithms for various models discussed in \cite{Stojnictcmspnncapliftedrdt23,Stojnicnflgscompyx23,Stojnicsflgscompyx23,Stojnicflrdt23} presents just a small subset of possibilities. These extensions usually require a bit of problem specific adjustment and we discuss them in separate papers.

\begin{singlespace}
\bibliographystyle{plain}
\bibliography{nflgscompyxRefs}
\end{singlespace}

\end{document}